\documentclass[english,11pt]{article}
\pdfoutput=1
\usepackage[a4paper]{geometry}
\usepackage{graphicx}
\usepackage[textwidth=8em,textsize=small]{todonotes}
\usepackage[sort,comma]{natbib}
\usepackage{amsmath,mathtools,bbm}
\usepackage{amsthm}
\usepackage[pagebackref=false,colorlinks]{hyperref}
\hypersetup{pdffitwindow=true,linkcolor=blue,citecolor=blue,urlcolor=blue}

\usepackage{etoolbox}
\makeatletter
\patchcmd{\@makecaption}
  {\parbox}
  {\advance\@tempdima-\fontdimen2} 
  {}{}
\makeatother  

\newtheorem{prop}{Proposition}

\renewenvironment{abstract}{%
  \noindent\textbf\abstractname .\hspace{1pt}
}{
  \endlist \par\bigskip\bigskip
}

\usepackage{lastpage}
\RequirePackage[hyperpageref]{backref}
\renewcommand*{\backref}[1]{} 
\renewcommand*{\backrefalt}[4]{
    \ifcase #1
       No referred.
    \or
       \emph{Referred to on page #2.}
    \else
       \emph{Referred to on pages #2.}
    \fi}

\usepackage[shortlabels]{enumitem}

\usepackage[titletoc,title]{appendix}

\binoppenalty=\maxdimen
\relpenalty=\maxdimen

\usepackage{algorithm}
\usepackage{algpseudocode}
\usepackage{bm}
\graphicspath{ {../figures/} }
\usepackage[caption = false]{subfig}

\begin{document}
\begin{center}
{\LARGE\bf Semi-Supervised Non-Parametric Bayesian Modelling of Spatial Proteomics}
\end{center}
\medskip
\begin{center}
{\large Oliver M. Crook$^{1,2}$, Kathryn S. Lilley$^{1}$, Laurent Gatto$^{3,\ast}$ and Paul D. W. Kirk$^{2,\ast}$ \\[15pt]
\emph{$^{1}$Cambridge Centre for Proteomics, Department of Biochemistry, University of Cambridge, U.K.}\\
\emph{$^{2}$MRC Biostatistics Unit, School of Clinical Medicine, University of Cambridge, U.K.}\\
\emph{$^{3}$de Duve Institute, UCLouvain, Belgium}\\
}
\end{center}

\bigskip

\begin{center}
Preprint, \today
\end{center}
\bigskip\bigskip

\begin{abstract}
Understanding sub-cellular protein localisation is an essential component to analyse context
specific protein function. Recent advances in quantitative mass-spectrometry (MS) have led to high resolution mapping of thousands of proteins to sub-cellular locations within the cell. Novel
modelling considerations to capture the complex nature of these data are thus necessary.
We approach analysis of spatial proteomics data in a non-parametric Bayesian framework, using
mixtures of Gaussian process regression models. The Gaussian process regression model accounts
for correlation structure within a sub-cellular niche, with each mixture component capturing the distinct correlation structure observed within each niche. Proteins with \textit{a priori} labelled locations motivate using  semi-supervised learning to inform the Gaussian process hyperparameters. We moreover provide an efficient Hamiltonian-within-Gibbs sampler for our model. As in other recent work, we reduce the computational burden
associated with inversion of covariance matrices by exploiting the structure in the covariance matrix. A tensor decomposition allows extended Trench and Durbin 
algorithms to be applied in order to reduce the computational complexity of covariance matrix inversion and hence accelerate computation.  A stand-alone R-package implementing these methods using high-performance C++ libraries is available at: \url{https://github.com/ococrook/toeplitz}
\end{abstract}

\section{Introduction}
For a protein to make appropriate interactions with binding partners and substrates, it must localise to the correct sub-cellular compartment \citep{Gibson:2009}. Furthermore, there is mounting evidence implicating aberrant protein localisation in disease, including cancer and obesity \citep{Olkkonen:2006, Laurila:2009, Luheshi:2008, De:2011, Cody:2013, Kau:2004, Rodriguez:2004, Latorre:2005, Shin:2013, Siljee:2018}. Mapping the sub-cellular location of proteins using high-resolution spatial proteomic approaches are thus of high utility in the characterisation of therapeutic targets and in determining pathobiological mechanisms \citep{Cook:2019}. To interrogate the sub-cellular locations of thousands of proteins per experiment, recent advances in high-throughput spatial proteomics
\citep{hyper, Mulvey:2017, DC:2018}, followed by rigorous data analysis \citep{Gatto:2010} can be applied.  The methodology relies on the observation that organelles, macro-molecular complexes and, more generally, sub-cellular niches are characterised by density-gradient profiles without the necessity to create homogeneous preparations of major sub-cellular components \citep{DeDuve:1981}.

Mass spectrometry(MS)-based spatial proteomics experiments begin with gentle lysis of the cell in such a way that maintains the integrity of their organelles. A diverse set of methods are available to
separate cellular content, including equilibrium density separation \citep{Dunkley:2004, Dunkley:2006,hyper} and
differential centrifugation \citep{Itzhak:2016, DC:2018, Orre:2019}, amongst others \citep{Parsons:2014, Heard:2015}. In the LOPIT
\citep{Dunkley:2004, Dunkley:2006, Sadowski:2006} and \textit{hyper}LOPIT \citep{hyper, Mulvey:2017} approaches, cell lysis is proceeded by the separation of
sub-cellular components along a continuous density gradient based on their buoyant density.
Discrete fractions along this gradient are then collected, and protein distributions revealing organelle specific correlation profiles within the fractions are achieved using high accuracy MS.

Sophisticated data analysis methods for spatial proteomics have been developed
\citep{Breckels:2013, Gatto:2014b, Breckels:2016, Crook::2018, Gatto:2019}, along with detailed work flows
\citep{ghrepo}, which have benefited from significant contributions to the R programming language \citep{R} and the Bioconductor project \citep{Bioconductor::2004, Huber::2015} through MS and proteomics packages \citep{MSnbase:2012, pRoloc:2014}.
Applications of such methods have led to organelle-specific localisation
information of proteins in many systems \citep{Dunkley:2006, Tan:2009, hall:2009, Breckels:2013}, including mouse pluripotent stem cells \citep{hyper} and cancer cell lines \citep{Thul:2017}. MS based spatial proteomics has gained in popularity in recent years with several recent applications across many different cell lines \citep{hyper, Beltran:2016, Jadot:2017, Itzhak:2017, Mendes:2017, Hirst:2018, Davies:2018, Orre:2019, Nightingale:2019}. This
motivates the development of a unified statistical framework for spatial proteomics.
Furthermore, with the absence of a mechanistic model for the data, quantifying uncertainty in systems biology applications is of paramount importance and, as yet, such a model has not be applied to existing datasets \citep{Kirk:2015}.

The current goal of computational methods is to assign proteins with unknown localisation to known sub-cellular niches. It is important to note, however, that not all proteins can be robustly assigned to single locations, since many proteins function in multiple cellular compartments, they may reside in uncharacterised organelles or they may translocate between multiple locations all leading to uncertainty in assignments. Recently, \cite{Crook::2018} demonstrated the importance of uncertainty quantification in spatial proteomics analysis. This study developed a generative model of these data and computed posterior distributions of protein localisation probabilities demonstrating the variety of reasons for uncertain protein localisation. This study, however, failed to model known features of the biochemical fractionation process. Sub-cellular niches
are not purified in a single fraction along the density-gradient but have distinct
quantitative profiles reflective of their biophysical properties. Since the exact quantitative profiles depends heavily on experimental design, the functions describing these profiles are unknown. This motivates a non-parametric Bayesian approach to analysing these data in order to infer the unknown functions and quantify the uncertainty in these functions.

We assume each quantitative protein profile can be described by some unknown function, with the uncertainty in this function captured using a \textit{Gaussian process (GP) prior}. Each sub-cellular niche is described by distinct density-gradient profiles, which display a non-linear structure with no particular parametric assumption being suitable. The contrasting density-gradient profiles are captured as components in a mixture of Gaussian process regression models. Gaussian process regression models have been applied extensively and we refer to \cite{Rasmussen:2004} and \cite{Rasmussen:2006} for the general theory. In molecular biology and functional genomics the focus of many applications has been on expression time-series data, where sophisticated models have been developed \citep{Kirk:2009, Cooke:2011, Kalaitzis:2011, Kirk:2012, Hensman:2013,Strauss:2019}. We remark that many of these applications consider unsupervised clustering problems. In contrast, here we have (partially) labelled data (proteins with  location known prior to our experiments) and so we may consider semi-supervised approaches. We explore inference of GP hyperparameters in two ways: firstly, an empirical Bayes approach in which the hyperparameters are optimised by maximising a marginal likelihood; secondly, by placing priors over these GP hyperparameters and performing fully Bayesian inference using labelled and unlabelled data.

A number of computational aspects need to be considered if inference is to be applied to spatial proteomics data. The first is that correlation in the GP hyperparameters can lead to slow exploration of the posterior, thus we use Hamiltonian evolutions to propose global moves through our probability space \citep{Duane:1987} avoiding random walk nature evident in
traditional symmetric random walk proposals \citep{Metropolis:1953, Beskos:2013}. Hamiltonian Monte-Carlo (HMC) has been explored previously for hyperparameter inference in GP regression \citep{Williams::1996}, and here we show that HMC can be up to an order of magnitude more efficient than a Metropolis-Hastings approach. Furthermore, a particular costly computation in our model is the computation of the marginal likelihood (and its gradient) associated with each mixture component, which involves the inversion of a large covariance matrix - even storage of such matrix can be challenging. We demonstrate that a tensor decomposition of the covariance matrix allows application of fast matrix algorithms for covariance inversion and low memory storage \citep{Zhang:2005}.

\section{Methods}
\subsection{Model specification}
In our experiment, we make discrete observations along a continuous density-gradient, ${\bf x}_i = \left[x_i(t_1),...,x_i(t_D)\right]$, where $x_i(t_j)$ indicates the measurement of protein $i$ at fraction $t_j$ along the gradient. We assume that protein intensity $x_i$ varies smoothly with the distance along the density-gradient. We further assume observations are equally spaced. Thus, the regression model for each protein is
\begin{equation}
x_i(t_j) = \mu_i(t_j) + \epsilon_{ij},
\end{equation}
where $\mu_i$ is an unknown deterministic function of space and $\epsilon_{ij}$ a noise variable. We assume that $\epsilon_{ij}\sim_{iid} \mathcal{N}(0, \sigma_{i}^2)$, for simplicity and remark that more elaborate noise models could be chosen but at additional computational cost and greater model complexity. Proteins are grouped together according to their sub-cellular localisation, with all proteins associated with sub-cellular niche $k = 1,...,K$ sharing the same regression model; that is,
$\mu_i = \mu_k$ and $\sigma_{i} = \sigma_k$. For clarity, we refer to sub-cellular structures, whether that be organelles, vesicles or large multi-protein complexes, as \textit{components}. Thus proteins associated with component $k$ can be modelled as \textit{i.i.d} draws from a multivariate Gaussian random variable with mean vector $\boldsymbol{\mu}_k = \left[\mu_k(t_1),...,\mu_k(t_D)\right]$ and covariance matrix $\sigma^2_kI_D$. To perform inference for unknown $\boldsymbol{\mu}_k$, as is typical for spatial correlated data \citep{Gelfand:2005, Steel:2010}, we specify a \textit{Gaussian Process prior} for each $\mu_k$
\begin{equation}
\mu_k \sim GP(m_k(t), C_k(t,t')).
\end{equation}
Each component is thus captured by a Gaussian process regression model and the full complement of proteins as a finite mixture of Gaussian process regression models.
\subsection{Finite mixture models}
This section provides a brief review of finite mixture models (see, for example \citep{Lavine:1992,Fraley:2007} for more details) . Finite mixture models are of the form,
\begin{equation}
p({\bf x}|\boldsymbol{\pi}, \boldsymbol{\theta}) = \sum_{k=1}^K \pi_k F({\bf x}|\boldsymbol{\theta}_k),\label{finite}
\end{equation}
where $K$ is the number of mixture components,  $\pi_k$ are the mixture proportions, and $F({\bf x}|\boldsymbol{\theta}_k)$ are the component densities.  We assume each component density to have the same parametric form, but with component specific  parameters, $\boldsymbol{\theta}_k$.
We denote the prior for these unknown component parameters by $G_0(\boldsymbol{\theta})$.  
We suppose that we have a collection of $n$ data points, $X = \{{\bf x}_1, \ldots, {\bf x}_n \}$ that we seek to model using Equation \eqref{finite}.  We associate with each of these data points a component indicator variable, $z_i \in \{1, \ldots, K\}$, which indicates which component generated  observation ${\bf x}_i$.  Given the mixing proportions, the joint prior distribution of these indicators is multinomial with parameter vector $\boldsymbol{\pi} = [\pi_1, \ldots, \pi_K]$,
\begin{equation}  
P(z_1, \ldots, z_n | \boldsymbol{\pi}) = \prod_{k=1}^K \pi_k^{n_k},
\end{equation}
where $n_k$ is the number of data points $x_i$ for which $z_i = k$.  If we assign the mixture proportions a symmetric Dirichlet prior with concentration parameter $\alpha/K$, then we may marginalise the $\pi_k$ in order to yield the following joint distribution for the indicators \citep{Murphy:2012},
\begin{equation}
P(z_1, \ldots, z_n | \alpha) = \frac{\Gamma({\alpha})}{\Gamma(n+\alpha)}\prod_{i=1}^K \frac{\Gamma(n_i + \alpha/K)}{\Gamma(\alpha/K)}.
\end{equation}
For Gibbs sampling, we require the conditional priors for a single indicator, $z_i$, given all of the others, $z_{-i}$.  These are given by \citep{Murphy:2012}, 
\begin{equation}
P(z_i = k | z_{-i}, \alpha) = \frac{n_{-i,k} + \alpha/K}{N - 1 + \alpha}, 
\end{equation}
where $n_{-i,k}$ is the number of observations, excluding ${\bf x}_i$, that are associated with component $k$.
If we are given the parameters, $\boldsymbol{\theta}_k$, associated with each of the components then we may combine the above conditional priors with the likelihoods, $F({\bf x}_i| {\bf \boldsymbol{\theta}}_k)$, in order to obtain the conditional posterior:
\begin{equation}\label{equation::conditionalposterior}
P(z_i = k | z_{-i}) \propto \frac{n_{-i,k} + \alpha/K}{N - 1 + \alpha} F({\bf x}_i| {\bf \boldsymbol{\theta}}_k).
\end{equation}
An alternative to integrating out the mixture proportions is to sample them directly from the posterior, which leads to increased posterior variance \citep{Gelfand:1990, Casella:1996} but can be computational advantageous. Conjugacy of the Dirichlet prior and multinomial likelihood means that the posterior distribution of the mixing proportions is also
Dirichlet,
\begin{equation}\label{equation::posteriormixing}
\boldsymbol{\pi}|z_1,...,z_n, \alpha \sim Dir(\alpha/K + n_1,..., \alpha/K + n_K).
\end{equation}
In this situation the conditional posterior becomes
\begin{equation}\label{equation::conditionalposteriorsample}
P(z_i = k |\boldsymbol{\pi}) \propto \pi_k F({\bf x}_i| {\bf \boldsymbol{\theta}}_k).
\end{equation}

\subsection{Gaussian Process priors}
A Gaussian Process (GP) is a continuous stochastic process such that any finite collection of these random variables is jointly Gaussian. A Gaussian Process prior is uniquely specified by a mean function $m$ and covariance function $C$, which determine the mean vectors and covariance matrices of the associated multivariate Gaussian distributions. To elaborate, assuming a GP prior for $\mu_k$ means that for indices $t_1,...,t_D$, the joint prior of $\boldsymbol{\mu}_k = \left[\mu_k(t_1),...,\mu_k(t_D)\right]^T$, is multivariate Gaussian with mean vector $\boldsymbol{m}_k = \left[m_k(t_1),...,m_k(t_D)\right]$ and covariance matrix $C_k(i,j) = C_k(t_i,t_j) $. Given no prior belief about symmetry or periodicity in our deterministic function, we assume our GP is centred with squared exponential covariance function
\begin{equation}
C_k(t_i,t_j) = a_k^2 \exp\left(- \frac{\lVert t_i - t_j \rVert^2_2}{l_k}\right).
\end{equation}
\subsection{Marginalising the unknown function}
Having adopted a GP prior with component specific parameters $a_k$ and $l_k$ for each unknown function $\mu_k$, we let observations associated with component $k$ be denoted by $X_k = \{x_1,...,x_{n_k}\}$. Our model tells us that
\begin{equation}
X_k |\boldsymbol{\mu}_k, \sigma_k \sim \mathcal{N}(\boldsymbol{\mu}_k, \sigma^2_kI_D).
\end{equation}
Then, we can write this as
\begin{equation}
\begin{split}
X_k(t_1),...,X_k(t_D) &|\mu_k, \sigma_k \sim
\\& \mathcal{N}(\mu_k(t_1),...,\mu_k(t_D),...,\mu_k(t_1),...,\mu_k(t_D), \sigma^2_kI_{n_kD}),
\end{split}
\end{equation}
where $\mu_k(t_1),...,\mu_k(t_D)$ is repeated $n_k$ times. Our GP prior tell us  
\begin{equation}
\mu_k(t_1),...,\mu_k(t_D),...,\mu_k(t_1),...,\mu_k(t_D)|a_{k}, l_{k} \sim \mathcal{N}(0, C_k),
\end{equation}
where $C_k$ is an  $n_kD\times n_kD$ matrix. This matrix is organised into $n_k \times n_k$ square blocks each of size $D$. The $(i,j)^{th}$ block of $C_k$ being $A_k$, where $A_k$ is the covariance function for the $k^{th}$ component evaluated at $\tau = \{t_1,...,t_D\}$.
\begin{equation}\label{equation::blockstructure}
C_k = 
\begin{bmatrix}
A_k & A_k & \dots & A_k \\
A_k & A_k & \dots & A_k \\
\vdots & \vdots & \ddots & \vdots \\
A_k & A_k & \dots & A_k
\end{bmatrix}.
\end{equation}
 Letting $\boldsymbol{\theta}_{k} = \left\{a_k, l_k, \sigma^2_k\right\}$, we can then marginalise $f_k$ to obtain,
\begin{equation}
\begin{split}
X_k(t_1),...,X_k(t_D) |\boldsymbol{\theta}_{k} \sim \mathcal{N}(0, C_k + \sigma^2_kI_{n_kD}),
\end{split}
\end{equation}
thus avoiding inference of $\mu_k$. Let  $X_k(\tau)$ denote the vector of length $n_k\times D$ equal to $\left(x_1(t_1),...,x_{1}(t_D),\dots, x_{n_k}(t_1),...,x_{n_k}(t_D)\right)$. Then we may rewrite equation \ref{equation::conditionalposterior} by marginalising $\mu_k$ to obtain:
\begin{equation}\label{equation::margconditionalposterior}
P(z_i = k | z_{-i}) \propto \frac{n_{-i,k} + \alpha/K}{K - 1 + \alpha} \int p({\bf x}_i|\mu_k)p(\mu_k| \boldsymbol{\theta_k},X_{-i,k}(\tau) ) \,\mathrm{d}\mu_k,
\end{equation}
where $X_{-i,k}(\tau)$ is equal to $X_k(\tau)$ with observation $i$ removed.
\subsection{Tensor decomposition of the covariance matrix for fast inference}
Our covariance matrix has a particularly simple structure allowing us to exploit extended Trench and Durbin algorithms for fast matrix computations \citep{Zhang:2005}. Recall  we are interested in the inversion of matrices of the following form
\begin{equation}\label{equation::covariancestructure}
C = 
\begin{bmatrix}
A + \sigma^2I_D & A & \dots & A \\
A & A + \sigma^2I_D & \dots & A \\
\vdots & \vdots & \ddots & \vdots \\
A & A & \dots & A + \sigma^2I_D
\end{bmatrix}.
\end{equation}
Note that $A$ is a positive symmetric matrix of size $D\times D$ and furthermore it is Toeplitz (constant diagonal and perisymmetric). Let $J_n$ denote an $n\times n$ matrix of ones. It is clear that we can write $C$ in the following form:
\begin{equation}
C = \sigma^2I_{nD} + B,
\end{equation}
where
\begin{equation}
B = J_{n} \otimes A,
\end{equation}
and $\otimes$ denotes the Kronecker (tensor) product. Let us write
\begin{equation}
Q = (I_D + \sigma^{-2} n A),
\end{equation}
which is a $D\times D$ Toeplitz matrix, for which the inverse and determinant can be inverted in $O(D^2)$  operations \citep{Durbin:1960, Trench:1964}. If we denote the inverse of $Q$ by $Z$ it follows (see supplementary for full derivation) that:
\begin{equation}
\begin{split}
C^{-1} &= \sigma^{-2}I_{nD} - \sigma^{-4}J_{n}^T \otimes (ZA)\\
& = \sigma^{-2}I_{nD} - \frac{1}{n\sigma^2}J_{n}^T \otimes (I - Z)
\end{split}
\end{equation}
and
\begin{equation}
\det(C) =  (\sigma^2)^{nD} \det(I_D + \sigma^{-2} n A).
\end{equation}
Thus, the inversion of $C$ requires only the inversion of a $D\times D$ matrix, which can be performed in $O(D^2)$ computations, this should be compared with a na\"ive inversion of $C$ requiring $O((nD)^3)$ computations, which represents significant savings. The determinant can also be obtained in $O(D^2)$ operations \citep{Zhang:2005}. The method of \cite{Zhang:2005} can be seen as a special case of our situation where $n = 1$. Step by step algorithms for computing this inverse and determinant can also be found in the supplementary materials.  We note that \citet{Strauss:2019} also exploit the block matrix structure of the covariance matrix efficiently,  using a more general approach to compute block matrix inversions and determinants that works also in the case of hierarchical GP models, for which \citet{Hensman:2013} had found an alternative way of performing efficient likelihood computations. 
   
\subsection{Sampling the underlying function}
Whilst it is often mathematically convenient to marginalise the unknown function $\mu_k$ from a computational perspective it is not always advantageous to do so. To be precise, marginalising $\mu_k$ induces dependencies among the observations; that is, we cannot exploit the conditional independence structure given the underlying function $\mu_k$. After marginalising, Gibbs moves must be made sequentially for each protein in turn and this can slow down computation. 

The alternative approach is to sample the underlying function and exploit conditional independence. Once a sample is obtained from the GP posterior on $\mu_k$, conditional independence allows us to compute the likelihood for all proteins at once, exploiting vectorisation. If there are a particularly large number of observation in each component it is also possible to parallelize computation over the components $k = 1,...,K$.  

\subsection{Modelling outliers}
\cite{Crook::2018} demonstrated that many proteins are not captured well by any known sub-cellular component. This could be because of yet undiscovered biological novelty, technical variation or a manifestation of some proteins residing in multiple localisations. Modelling outliers in mixture models can be challenging \citep{Hennig::2004,Cooke:2011,Coretto::2016}. Here, we take the approach of \cite{Crook::2018}. Briefly, we introduce a further binary latent variable $\phi$ so that for each protein ${\bf x}_i$ we have a $\phi_i$ indicating whether ${\bf x}_i$ is modelled by one of the known components or an outlier component. The augmented model becomes the following
\begin{equation}
p({\bf x}_i|\boldsymbol{\pi}, \boldsymbol{\theta}) = \sum_{k=1}^K \pi_k F({\bf x}_i|\boldsymbol{\theta}_k)^{\phi_i} G(x_i|\Phi)^{1 - \phi_i},\label{augfinite}
\end{equation}
where $G$ is density of the outlier component. In our case, we specify $G$ as the density of a multivariate T distribution
with degrees of freedom $\kappa = 4$, mean ${\bf M}$ and scale matrix $V$. ${\bf M}$ is taken as the empirical global mean of the data and the scale matrix $V$ as half the empirical covariance of the data. These choices are motivated by considering a Gaussian component with the same mean and covariance but with heavier tails to better capture dispersed proteins. We remark that other choices of G and parameters may be suitable and can be tailored to the application at hand. In typical Bayesian fashion, we specify a prior for $\phi$ as $p_0(\phi_i = 0) = \epsilon$, where $\epsilon \sim \mathcal{B}(u,v)$. All hyperparameter choices are stated in the appendix.
\subsection{Gaussian process hyperparameter inference}
\subsubsection{Supervised approach: optimising the hyperparameters}
Inference of the hyperparameters $\boldsymbol{\theta_{k}}$ can be dealt with in several ways. The first is to learn them using only the labelled data (i.e. data that pertains to proteins with well documented sub-cellular locations). Using the labelled data for each component constitutes maximise the marginal likelihood of the hyperparameters with respect to the data. These hyperparameters are then fixed throughout the inference of the unlabelled data. The marginal likelihood can be obtained quickly by recalling that
\begin{equation}
\begin{split}
X_k(t_1),...,X_k(t_D) |\boldsymbol{\theta_{k}} \sim N(0, C_k + \sigma^2_kI_{n_kD}).
\end{split}
\end{equation} 
Thus the log marginal likelihood is given by
\begin{equation}\label{equation::GPmarginalliklihood}
\begin{split}
\log p(X_k|\tau,& \boldsymbol{\theta_{k}}) \\
& = -\frac{1}{2}X_k(\tau)\left(C_k + \sigma^2_kI_{n_kD}\right)^{-1}X_k(\tau)^{T} - \frac{1}{2} \log \lvert C_k + \sigma^2_kI_{n_kD}\rvert - \frac{n_kD}{2} \log 2\pi.
\end{split}
\end{equation}
For convenience of notation set $\hat{C}_k = C_k + \sigma^2_kI_{n_kD}$. To maximise the marginal likelihood given equation \ref{equation::GPmarginalliklihood}, we find the partial derivatives with respect to the parameters \citep{Rasmussen:2004}. Hence, we can use a gradient based optimisation procedure. Positivity constraints on $a^2_k,l_k,\sigma_k^2$ are dealt with by re-parametrisation and so, dropping the dependence on $k$ for notational convenience, and abusing notation, we set $l = \exp(\theta_1)$, $a^2 = \exp(2\theta_2)$ and $\sigma^2 = \exp(2\theta_3)$.
\\
\\
Application of the quasi-Newton L-BFGS algorithm \citep{Liu:1989} for numerical optimisation of the marginal likelihood with respect to the hyperparameters is now straightforward. The L-BFGS can only find a local optimum and so we initialise over a grid of values. We terminate the algorithm when successive iterations of the gradient are less than $10^{-8}$. We make extensive use of high performance R packages to interface with C++ \citep{Rcpp:2011,arma:2014}.
\subsubsection{Semi-supervised hyperparameter inference}
The advantage of adopting a Bayesian approach to hyperparameter inference is that we can quantify uncertainty in these hyperparameters. Uncertainty quantification in GP  hyperparameter inference is important, since different hyperparameters can have a strong effect on the GP posterior \citep{Rasmussen:2004}. Furthermore, we consider a semi-supervised approach to hyperparameter inference. By a semi-supervised approach we mean that a posterior distribution for the hyperparameters can be inferred using both the labelled and unlabelled data, rather than just the labelled data.
\\
\\
Consider at some iteration of our MCMC algorithm the data associated to the $k^{th}$ component $X_k$. We can partition this data into the unlabelled (U) and labelled data (L); in particular, $X_k = \left[X^{(L)}_k, X^{(U)}_k\right]$. To clarify, the indicators $z_i$ are known for $X^{(L)}_k$ prior to any inference, whilst allocations $z_i$ for $X^{(U)}_k$ are sampled at each iteration of our MCMC algorithm. If we believe our labelled data $X^{(L)}_k$ are true representatives of the distribution of that component, it is computationally advantageous just to consider the labelled data when performing hyperparameter inference. However, there could be a sampling bias in the labelled data and so the labelled data alone is insufficient to explain the variability in the data. A semi-supervised approach allows the posterior distribution of the hyperparamters to reflect the uncertainty in the component allocations $z_i$ and therefore improve our abilities to predict allocations and quantify uncertainty in allocations.  
\subsubsection{Semi-Supervised approach: hyperparameter inference using MH}
In a Bayesian framework, we treat the hyperparameters as random variables and place hyperpriors overs them. Positivity constraints  motivate working with the $\log$ of the hyperparameters and using, for example, standard normal priors \citep{Neal:1997}. Unfortunately loss of conjugacy between the prior on the hyperparameters $G_0(\boldsymbol{\theta})$ and the likelihood $F({\bf x}| \boldsymbol{\theta})$ is unavoidable, and hence we use a Metropolis-Hastings step or Hamiltonian Monte-Carlo step for inference. The Metropolis-Hastings sampler can be summarised as follows:
\\
\\
Metropolis-Hastings algorithm with random walk proposals: Suppose $\theta_t$ is the most recently sampled value. Sample a value $\xi \sim N(0,1)$, setting $\theta_{t+1} = \theta_{t} + \xi$  and compute the Metropolis ratio 
\begin{equation}
\Lambda = \frac{p(\theta_{t+1}|X_k(\tau))}{p(\theta_{t}|X_k(\tau))} = \frac{p(X_k(\tau)|\theta_{t+1})p_0(\theta_{t+1})}{p(X_k(\tau)|\theta_{t})p_0(\theta_{t})}.
\end{equation}
This ratio can be computed in $\log$ form using equation $\ref{equation::GPmarginalliklihood}$. Then sample a uniform random number $u\sim U[0,1]$  if $\log(\Lambda) \geq \log(u)$ set $\theta_{t+1} =\theta_t + \xi$, otherwise $\theta_{t+1} =\theta_t$.
\subsubsection{Semi-Supervised approach: hyperparamter inference using HMC}
To avoid the random walk nature of the MH sampler, we also consider a Hamiltonian Monte-Carlo approach, which exploits the geometry of the space to provide more efficient proposals \citep{Duane:1987,Horowitz:1991, Neal:2011,Girolami::2011}. In short, Hamiltonian Monte-Carlo allows us to construct Hamiltonian evolutions $H(\bm{x},\bm{p})$ such that the resulting dynamics efficiently explore a target distribution $p(\bm{x})$. We augment our probability distribution with an auxiliary momentum component $\bm{p}$. An MCMC algorithm can then be constructed to sample from the required distribution, where proposals are made using Hamiltonian evolutions. Full details in the case of the hyperparameters of a Gaussian process are discussed in the supplement. In previous sections, we saw we can exploit a tensor decomposition to accelerate computation of the likelihood and similar formulae are available to accelerate computation of the gradient for use in L-BFGS and Hamiltonian Monte Carlo. These formulae can be found in the supplement.
\subsubsection{An overview of the MCMC algorithm for posterior Bayesian computation}
In our model $G_0(\boldsymbol{\theta})$ and $F({\bf x}| \boldsymbol{\theta})$ are non-conjugate, which means the integral in equation \ref{equation::margconditionalposterior} cannot be obtained analytically. A Gibbs sampling scheme with either an additional Metropolos-Hastings or Hamiltonian Monte Carlo update is used. Each iteration of the MCMC algorithm includes a sampled value for the component indicators, outlier components and current values of the hyperparameters. We also keep track of associated posterior probabilities and marginal likelihoods as appropriate. Furthermore, we can sample the hyperparameters every $T$ iterations of the MCMC algorithm to accelerate computations.
\subsection{Summarising uncertainty in posterior localisation probabilities}
Summarising uncertainty quantified by Bayesian analysis in an interpretable way can be challenging. As always, we can summarise uncertainty using credible intervals or regions \citep{Gelman:1995}. One particularly challenging quantity of interest to summarise is the uncertainty in posterior allocations. Whilst, each individual allocation of a protein to a sub-cellular niche can be summarised by a credible interval it is not clear what is the best way to summarise the posterior  over all possible localisations for each individual protein. As in previous work \citep{Crook::2018}, we propose to summarise this uncertainty in an information-theoretic approach by computing the Shannon entropy of the localisation probabilities \citep{shannon:1948} at each iteration of the MCMC algorithm
\begin{equation}
\left\{ H_{ik}^{(t)} = - \sum_{k=1}^Kp_{ik}^{(t)} \log p_{ik}^{(t)} \right\}_{t=1}^T,
\end{equation}
where $p_{ik}^{(t)}$ is the probability that protein $i$ belong to component $k$ at iteration $t$. We can then summarise this by a Monte-Carlo average:
\begin{equation}
H_{ik} \approx \frac{1}{T} \sum_{t=1}^{T} H_{ik}^{(t)}.
\end{equation}
We note that, larger values of a Shannon entropy correspond to greater uncertainty in allocations.
\section{Results}
\subsection{Case Study I: \textit{Drosophila melanogaster} embryos}
\subsubsection{Application}
The first case study is the \textit{Drosophila melanogaster} (common fruit fly) embryos \citep{Tan:2009}, in which we compare the supervised and semi-supervised approaches for updating the model hyperparameters. In particular, we explore the effect on the component specific noise term $\sigma^2$, by adopting different inference approaches. For each sub-cellular niche, we learn the hyperparameters by either maximising their marginal likelihood or sampling from their posterior using MCMC. The posterior distribution for the hyperparameters can either be found solely using the labelled data for each component or by making use of labelled and unlabelled data.

Figure \ref{fig:PosteriorTan} demonstrates several phenomena. Reassuringly, the estimates of the noise parameters $\sigma_k^2$ for $k = 1,...,K$ obtained by using the L-BFGS algorithm to maximise the marginal likelihood coincide  with the posterior distributions of the noise parameters, inferred using only the labelled data for each component. However, when we perform inference in a semi-supervised way, by using both the labelled and unlabelled data to make inferences, we make several important observations.

Firstly, in many cases, the posterior using both the labelled and unlabelled data is shifted right towards $0$. Recalling that we are working with the log of the hyperparameters, this indicates that the noise parameters is smaller when solely using the labelled data. This is likely a manifestation of experimental bias, since it is reasonable to believe that proteins with known prior locations are those which have less variable localisations and are therefore easier to experimentally validate. A semi-supervised approach is able to overcome these issues, by adapting to proteins in a dense region of space. In some cases the shift is pronounced, with posteriors of the parameters using labelled and unlabelled data found in the tails of the posterior only using the labelled distribution. Furthermore, we notice shrinkage in the posterior distribution of the noise parameter in the semi-supervised setting. The reduction in variance reduces our uncertainty about the underlying true value of $\sigma_k^2$ for $k = 1,...,K$. This variance reduction is observed in most cases even when these is little difference in the mean of the posteriors.

The primary goal of spatial proteomics is to predict the localisation of unknown proteins from data. Our modelling approach allows the allocation probability of each protein to each component to be used to predict the localisation of unknown proteins. Proteins may reside in multiple locations and some sub-cellular niches are
challenging to separate because of confounding biochemical properties, leading to uncertainty in a proteins localisation. Thus adopting a Bayesian approach and quantifying this uncertainty is of great importance. Our methods allow point-estimates as well as interval estimates to be obtained for the posterior localisation probabilities. Figure \ref{fig:pcaTan} demonstrates the results of applying our method. Each protein in this PCA plot is scaled according to mean of the Monte-Carlo samples from the posterior localisation probability.

Further visualisation of the model and data are possible. We plot two representative example of gradient-density profiles for two components the endoplasmic reticulum (ER) and the nucleus, in figure \ref{fig:gdTan}. We plot
both the labelled proteins, in colour, which were assigned to each component before our analysis. In grey, for both components, we plot the unlabelled proteins which have been allocated to these components probabilistically. We observe that they have the same gradient-density shape as the labelled proteins - in line with our beliefs about the underlying biology: that proteins from the same components should co-fractionate  and therefore have similar density gradient profiles. In addition, we overlay the posterior predictive distribution for these components and observe they represent the data well.
\begin{figure}[h]
	\centering
	\includegraphics[width= 12cm, height = 8cm]{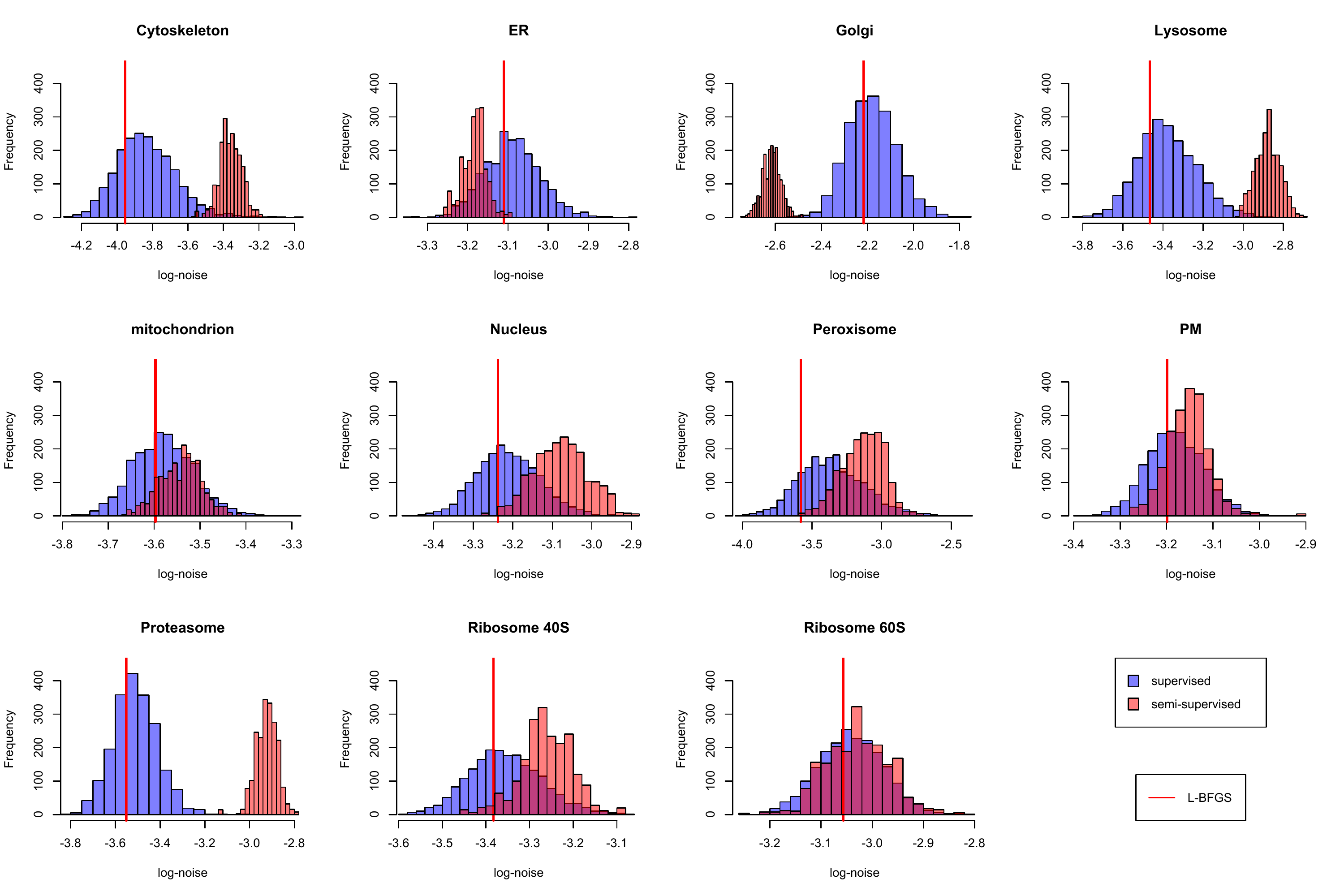}
	\centering
	\caption{Posterior distributions for the log noise parameter $\sigma^2$ on the
		\textit{Drosophila} data. In general, we observe a shift towards $0$, indicating that the labelled data underestimates the value of the noise term for each component. We also observe increased posterior shrinkage for many components with the variance of the noise parameters reduced in the semi-supervised setting.
	}
	\label{fig:PosteriorTan}
\end{figure}
\begin{figure}[h]
	\centering
	\includegraphics[width= 14cm, height = 11cm]{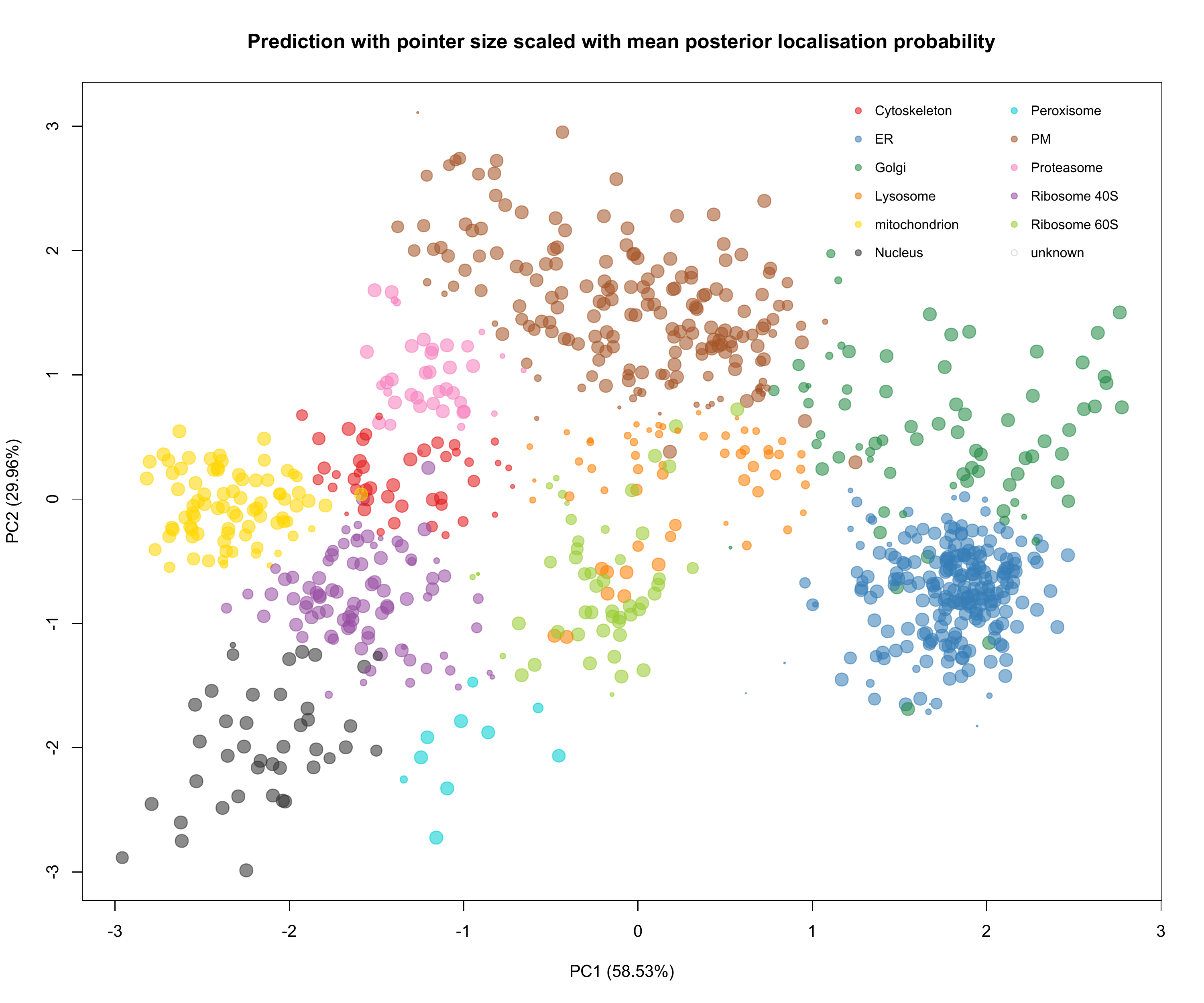}
	\centering
	\caption{A pca plot for the \textit{Drosophila} data where points, representing proteins, are coloured by the component of greatest probability. The pointer for each protein is scaled according to  membership probability with larger/smaller points indicating greater/lower allocation probabilities.
	}
	\label{fig:pcaTan}
\end{figure}
\begin{figure}[ht]
	\centering
	\begin{subfloat}[]{
		\includegraphics[height=1.2in]{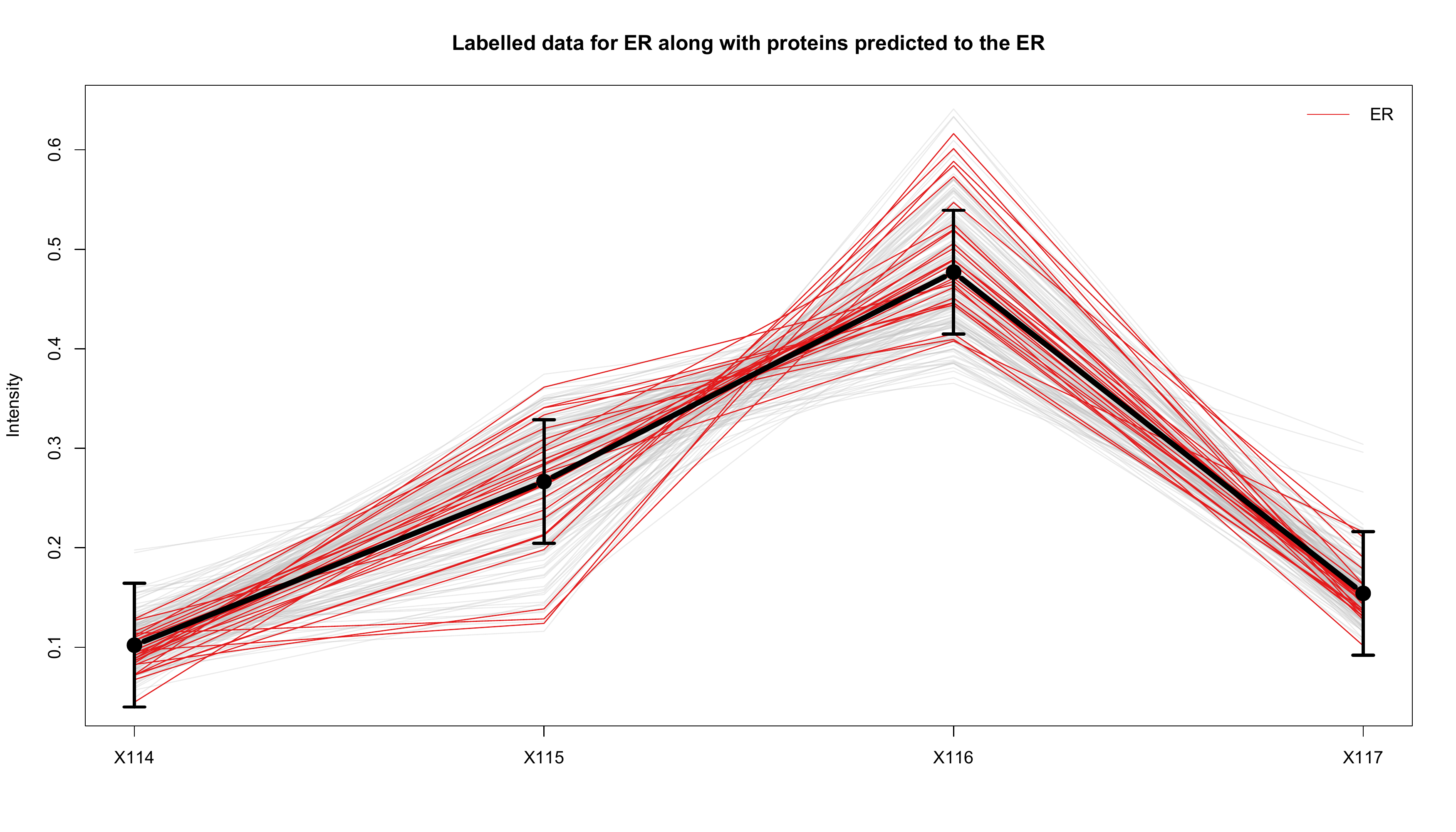}}
	\end{subfloat}%
	~
	\begin{subfloat}[]{
		\includegraphics[height=1.2in]{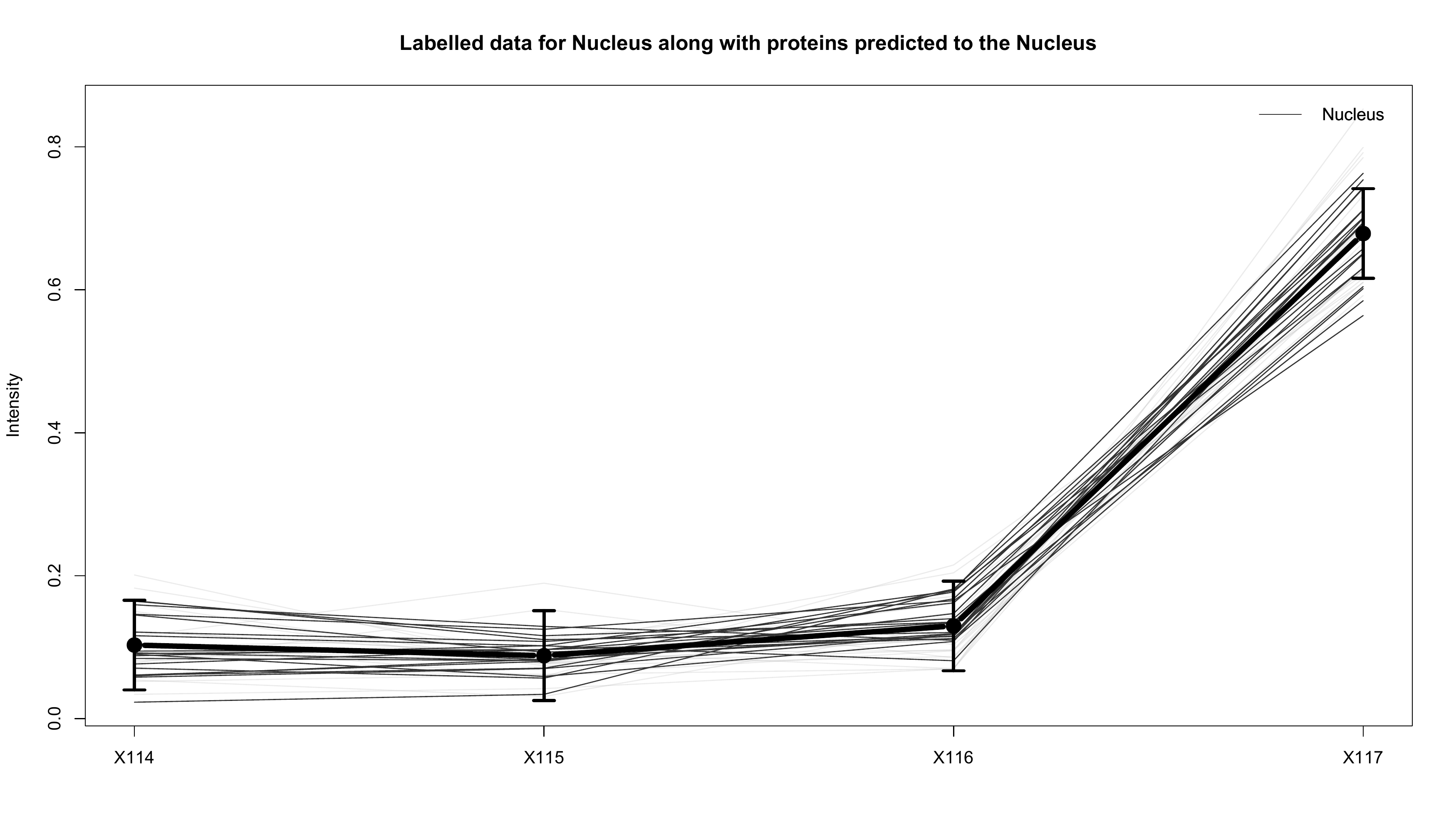}}
	\end{subfloat}%
	\caption{A plot of the gradient-density profiles for the ER and Nucleus with labelled proteins in colour and protein probabilistically assigned to those components in grey. The profiles of the assigned proteins closely match the profiles of the components. The predictive posterior of these components is also overlayed}
	\label{fig:gdTan}
\end{figure}
\clearpage
\newpage
\subsubsection{Sensitivity analysis for hyper-prior specification}
 We use the \textit{Drosophila melanogaster} dataset to test for sensitivity of the hyper-prior specification. To test for sensitivity, we see if predictive performance is affected by changes in the choice of hyper-prior. The following cross-validation schema assesses whether predictive performance is affected by choice of hyper-prior. We split the labelled data for each experiment into class-stratified  training $(80\%)$ and test $(20\%)$ partitions, with the separation formed at random. The true classes of the test profiles
are withheld from the classifier, whilst MCMC is performed. This $80/20$ data stratification is performed $100$ times in order produce a distribution of scores.
We compare the ability of the methods to probabilistically infer the true classes using the quadratic loss, also referred to as the Brier score \citep{Gneiting:2007}. Thus a distribution of quadratic losses is obtained for each method, with the preferred method minimising the quadratic loss. Each
method is run for $10,000$ MCMC iterations with $1000$ iterations for burn-in. We vary the mean of the standard normal hyper-prior for each hyperparameter in turn for a grid of values $\tilde{m}= (0, -1, -2, -3, -4)$, keeping the hyper-prior for the other variable held the same as a standard normal distribution. The results are displayed in figure \ref{SensitivityAnalysis}.

\begin{figure}[h]
	\includegraphics[width=12cm, height = 8cm]{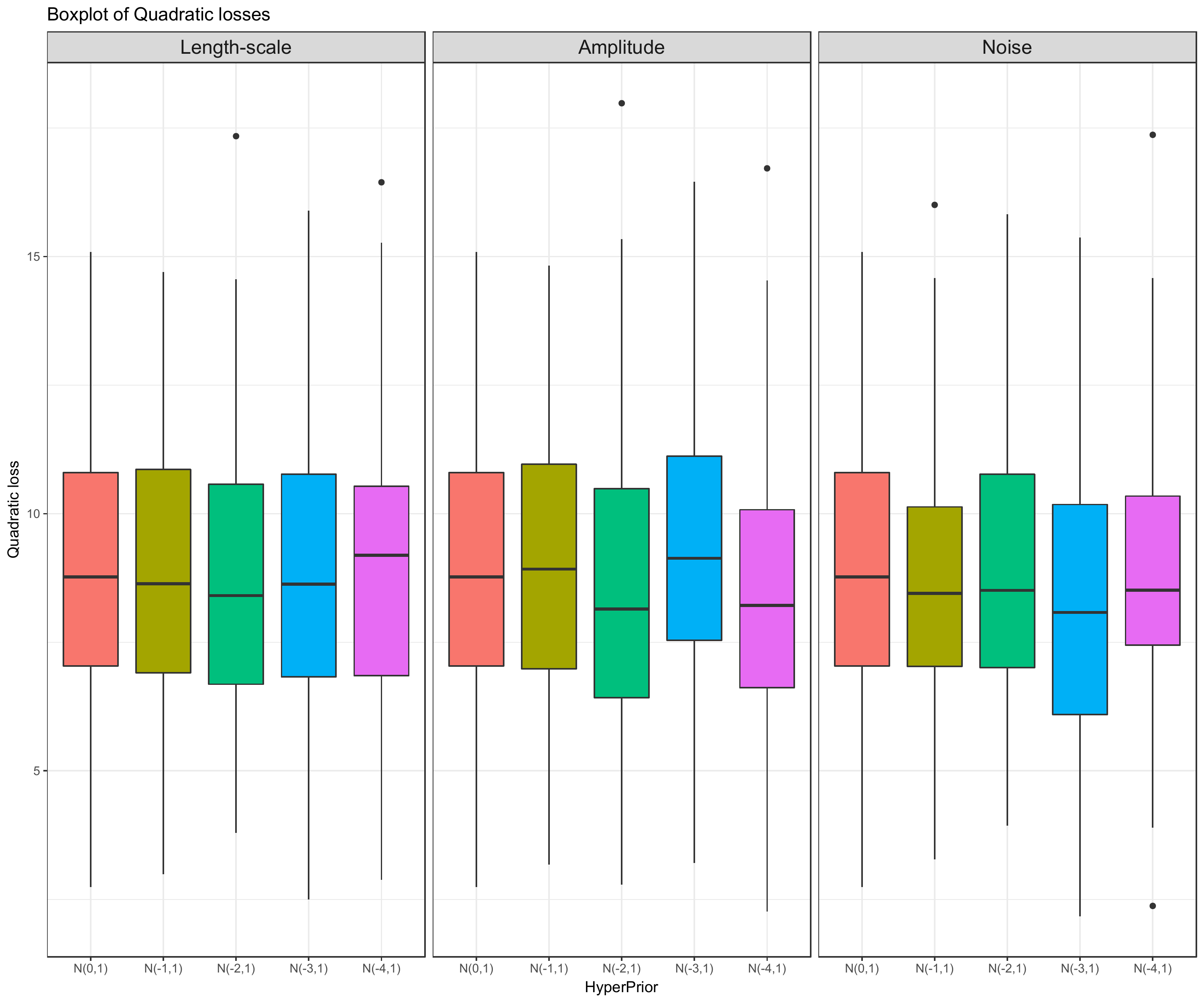}
	\centering
	\caption{Boxplots of quadratic losses to asses the sensitivity of semi-supervised hyperparameter inference to hyper-prior choices. 
	}
	\label{SensitivityAnalysis}
\end{figure}

We observe only minor sensitivity to the choice of hyper-prior, with no significant difference in performance noted (KS test, threshold = $0.01$). Sensitivity analysis for hyperparameters of GPs is vital, since these hyperparameters have a strong effect on the posterior of the GP \citep{Rasmussen:2004}. The observed lack of sensitivity in our case is advantageous, since prior information can be included without fear of over fitting. However, practitioners should always take care when specifying priors, especially for variance/covariance parameters as many authors have noted sensitivity of Bayesian models to these parameters \citep{Gelman:1995,Lunn:2000,Gelman::2006, Wang:2011,Schuurman:2016}

\subsection{Case Study II: mouse pluripotent embryonic stems cells} \label{section:mouse}
\subsubsection{Application}
Our main case study is the mouse pluripotent E14TG2a stem cell dataset of \cite{hyper}. This dataset contains 5032 quantitative protein profiles, and resolves 14 sub-cellular niches. We first plot the density-gradient profiles of the marker proteins for each sub-cellular niche in figure \ref{fig:GPallprofiles}. We fit a Gaussian process prior regression model for each sub-cellular niche with the hyperparameters found by maximising the marginal likelihood.
\begin{figure}[h]
	\includegraphics[width=12cm]{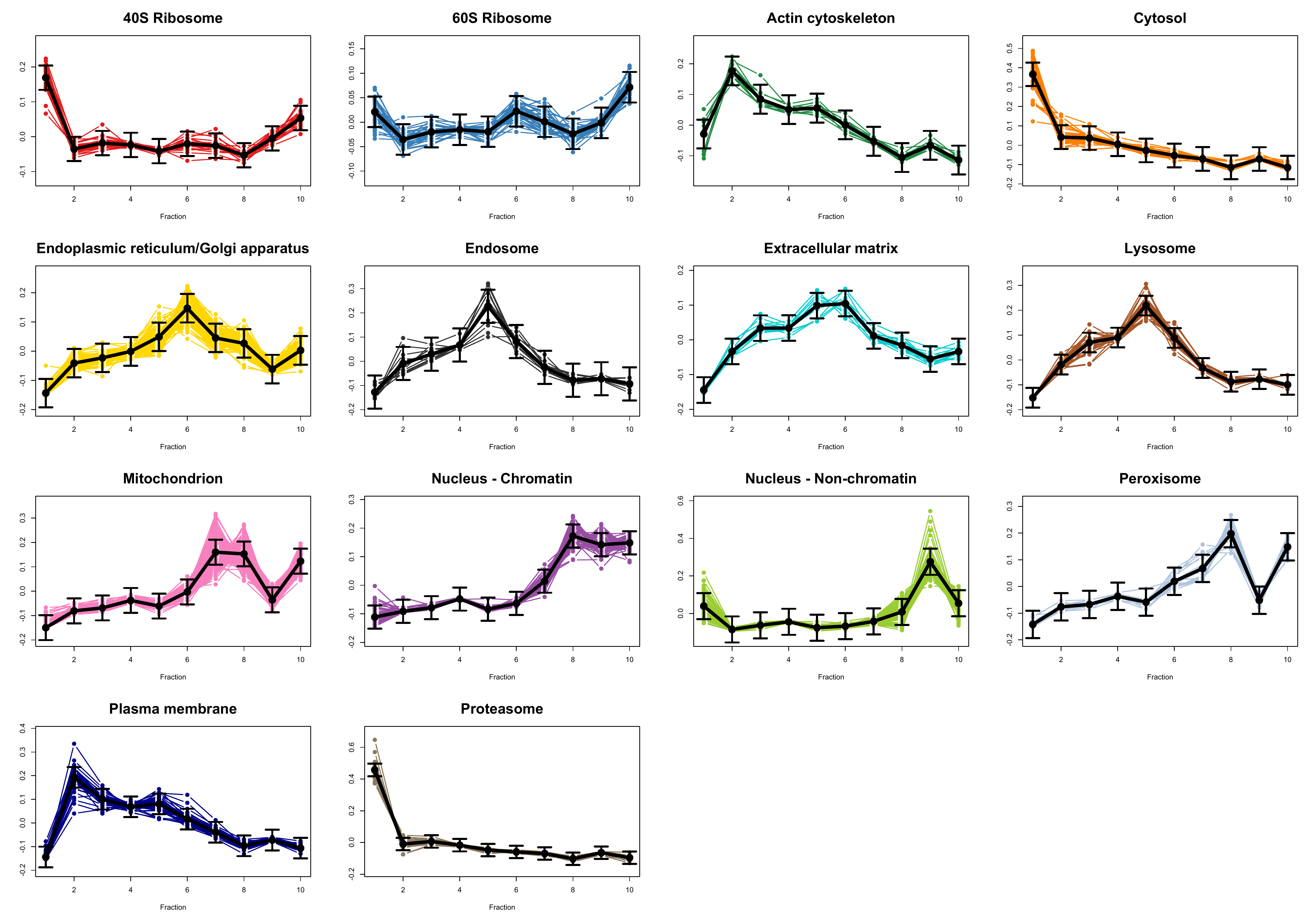}
	\centering
	\caption{Quantitative profiles of protein markers for each sub-cellular niche. A GP prior regression model is fitted to these data and the predictive distribution is displayed. We observe distinct distributions for each sub-cellular niche generated by the unique density-gradient properties of each sub-cellular niche.
	}
	\label{fig:GPallprofiles}
\end{figure}
A table of unconstrained log hyperparameter values found by maximising the marginal likelihood is found in the supplement. Alternatively, placing standard normal priors on each of the log hyperparameters and using a Metropolis-Hastings update we can infer the distributions over these hyperparameters. We perform $20,000$ iterations for each sub cellular niche and discard $15,000$ iterations for burn-in and proceed to thin the remaining samples by $20$. We summarise the
Monte-Carlo sample by the expected value as well as the $95\%$ equi-tailed credible interval, which can also be found in the supplement.

We go further to predict proteins with unknown localisation to annotated components using our proposed mixture of GP regression models. As before, we adopt a semi-supervised approach to hyperparameter inference. Again we place standard normal hyper-priors on the log of the hyperparameters. We run our MCMC algorithm for $20,000$ iterations with half taken as burnin and thin by $5$, as well as using HMC to update the hyperparameters. The PCA plot in figure \ref{fig:andyPCA} visualises our results. Each pointer represent a single protein and is scaled either to the probability of membership to the coloured component (left) or scaled with the Shannon entropy (right). In these plots we observe regions of high-probability and confidence to each organelle, as well as obtaining a global view of uncertainty. In this example, we observe regions of uncertainty, as measured by the Shannon entropy, concentrating where components overlap. We also 
observe uncertainty in regions where there is no dominant component. This Bayesian analysis provides a wealth of information on the global patterns of protein localisation in mouse pluripotent embryonic stem cells.
\begin{figure}[h]
	\centering
	\includegraphics[width= 15cm, height = 10cm]{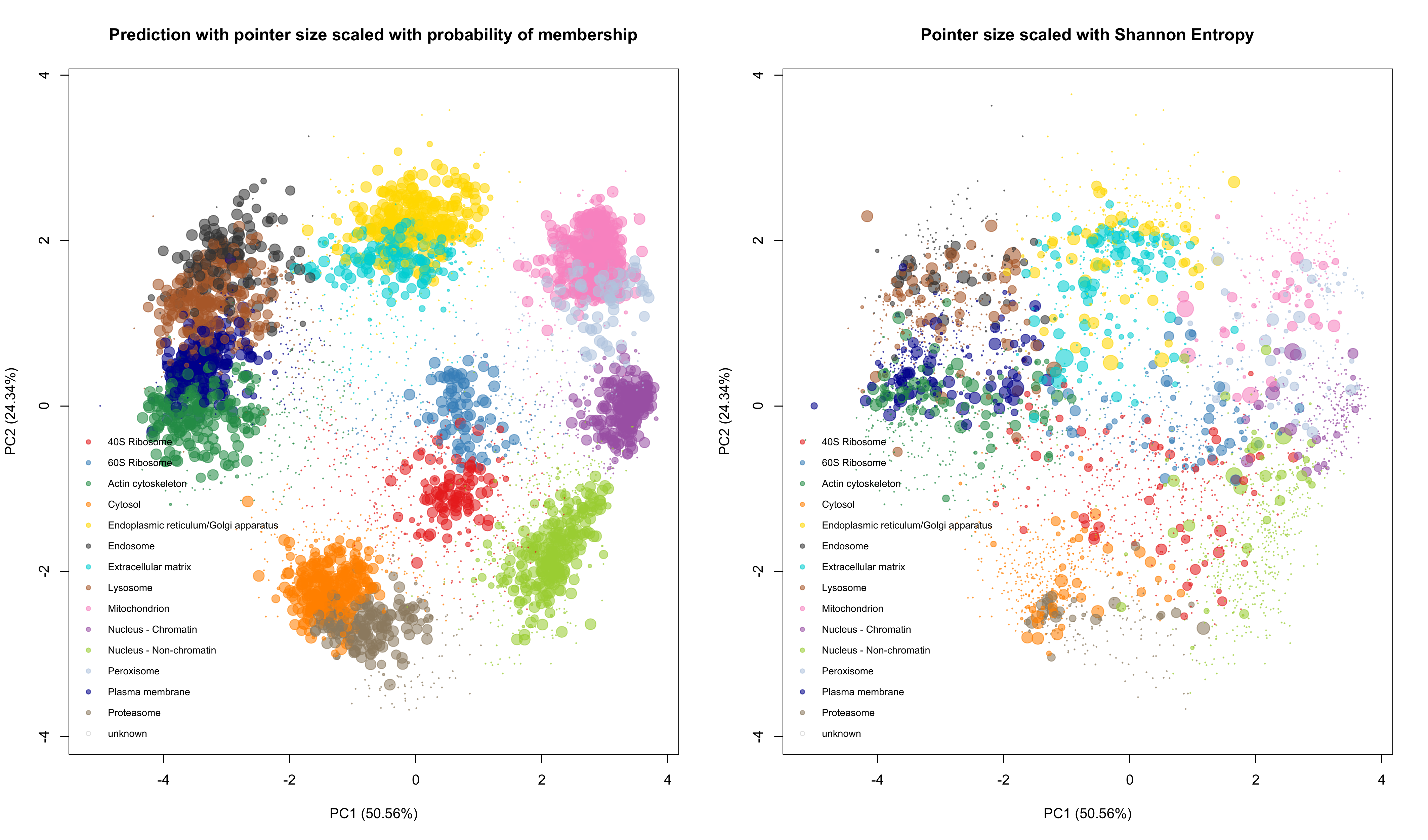}
	\centering
	\caption{A pca plot for the mouse pluripotent embryonic stem cell data where points, representing proteins, are coloured by the component of greatest probability. The pointer for each protein is scaled with membership probability (left). The pointer for each protein is scaled with the Monte-Carlo averaged Shannon Entropy.
	}
	\label{fig:andyPCA}
\end{figure}
\begin{table}
	\caption{A table summarising the difference in performance between Metropolis-Hastings and Hamiltonian Monte Carlo at sampling the 
	hyperparameters of a GP for several different organelles. For each organelle and for each method we report the acceptance rate and the time-normalised effective sample size. It is clear that HMC outperforms MH according to this metric.}
	\centering
	\fbox{%
    \begin{tabular}{c |c|c c c c c}
		\hline
		Component &	Method & Iterations & Acceptance  & Length-scale & Amplitude & Noise \\
		                    &	              &                &  rate             &  &  &  \\
		\hline
		Cytosol & MH & 50,000 & 0.240& 523 & 659 & 9375\\
		
		& HMC& 500 & 0.716&35348& 54730&  134485\\
		\hline
		Ribosome 40S & MH & 50,000 & 0.297&259 & 582& 10756 \\
		
		& HMC & 500 & 0.742&14114& 44662& 27758 \\
		\hline
		Lysosome & MH & 50,000 & 0.273&403& 821 &10385 \\
		
		& HMC & 500 & 0.710&28558& 40955& 543828\\
		\hline
		Proteosome  & MH & 50,000 & 0.267&408& 712& 10410 \\
		
		& HMC & 500 & 0.800&16243 &27186&  55923\\
		
		\hline
		Actin  & MH & 50,000 & 0.409&436 &1129 & 10841  \\
		
		& HMC & 500 & 0.598&5750&   479&   6342
	\end{tabular}
}
\end{table}

\subsection{Assessing predictive performance}\label{section:performance}
We compare the predictive performance of the methods proposed here,
as well as against the fully Bayesian TAGM model of \cite{Crook::2018}, where
sub-cellular niches are described by multivariate Gaussian distributions
rather than GPs. The following cross-validation schema is used to compare the classifiers. We split the labelled data for each experiment into class-stratified 
training $(80\%)$ and test $(20\%)$ partitions, with
the separation formed at random. The true classes of the test profiles
are withheld from the classifier, whilst MCMC is performed. This $80/20$ data stratification is performed $100$ times in order produce a distribution of scores.
We compare the ability of the methods to probabilistically infer the true classes using the quadratic loss, also referred to as the Brier score \citep{Gneiting:2007}. Thus a distribution of quadratic losses is obtained for each method, with the preferred method minimising the quadratic loss. Each
method is run for $10,000$ MCMC iterations with $1000$ iterations for burn-in. For fair comparison we held priors the same across all datasets. Prior specifications are stated in the supplement.

We compare across $5$ different spatial proteomics datasets across three different
organisms. The datasets we compare our methods on are \textit{Drosophila melanogaster} embryos from \cite{Tan:2009}, the mouse pluripotent embroyonic stem cell dataset of \cite{hyper}, the HeLa cell line dataset of \cite{Itzhak:2016}, the mouse primary neuron dataset of \cite{Itzhak:2017} and finally a CRISPR-CAS9 knock-out coupled to spatial proteomics analysis dataset (AP5Z1-KO1) of \cite{Hirst:2018}. The results are found in figure \ref{fig:compareTAGMGP}.
\begin{figure}[h]
	\includegraphics[width=12cm]{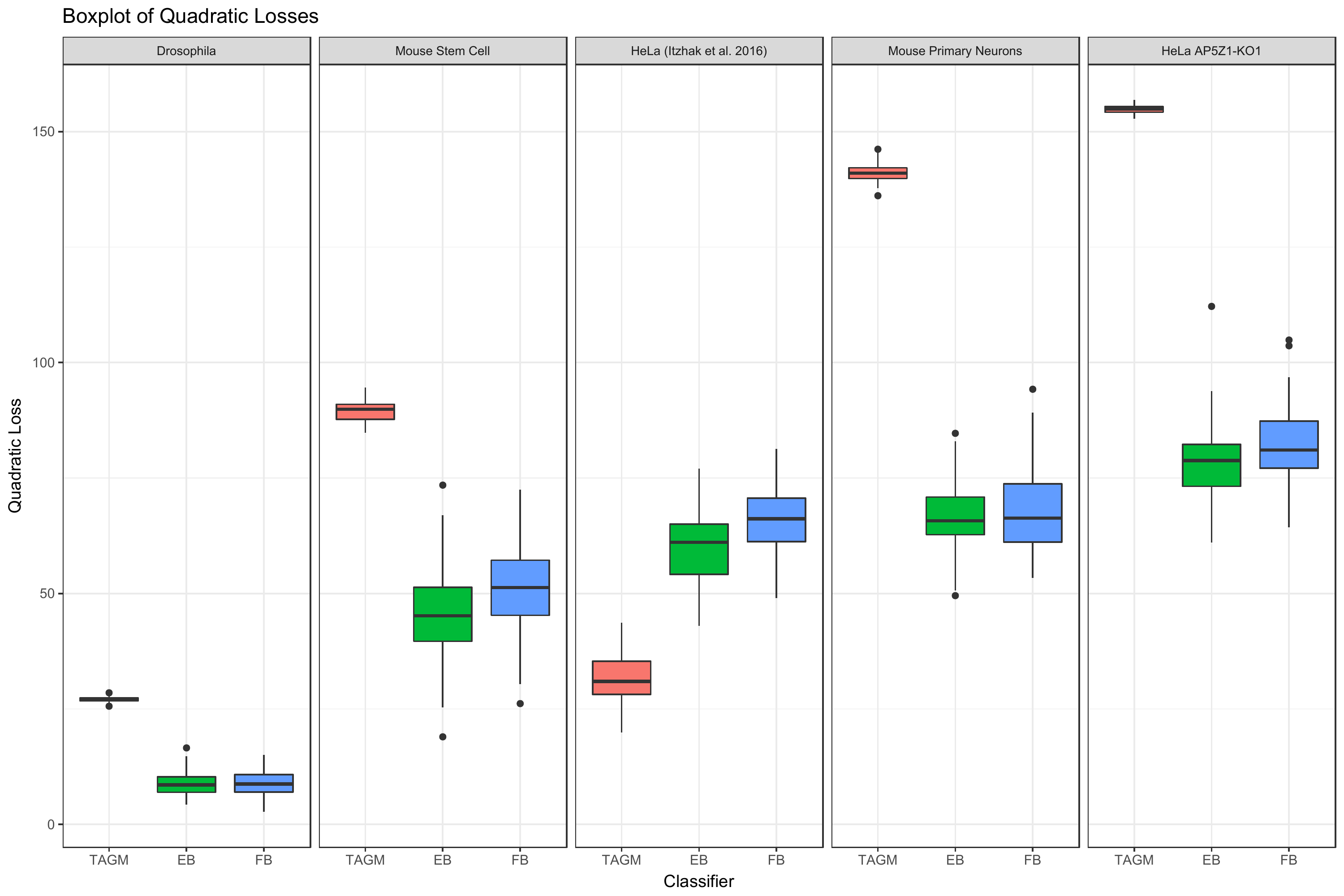}
	\centering
	\caption{Boxplots of quadratic losses comparing predictive performance of 
		the TAGM against the two semi-supervised Gaussian process models described here, where either an empirical Bayes (EB) approach or
		fully Bayesian (FB) approach is used for hyperparameter inference. 
	}
	\label{fig:compareTAGMGP}
\end{figure}

We see that our in four out five datasets there is an improvement of the GP models
over the TAGM model (Kolmogorov-Smirnov (KS) two-sample test $p < 0.0001$ ), because the GP model is provided with more explicit correlation structure of the data. The empirical Bayes slightly method outperforms the fully Bayesian approach in three of the data sets ((KS) two-sample test $p < 0.01$ ). These are the mouse pluripotent embryonic stem cell dataset, the HeLa data set of \cite{Itzhak:2016} and the HeLA AP5Z1 knock-out dataset of \cite{Hirst:2018}. We observe, the TAGM model outperforms the GP methods in the \cite{Itzhak:2016} dataset. The authors of this
study used differential centrifugation to separate cellular content and curated a ``large protein complex" class. This class could contain multiple sub-cellular structures such as ribosomes, as well as cytosolic and nuclear proteins. In any case, our modelling assumptions are violated in both models and this is issue
is exacerbated by parametrising the covariance structure. One solution to this would be to model this mixture of large protein complexes as its own class. However, as this class contains a quite diverse set of sub-cellular compartments, it is difficult to predict behaviour. This class could be itself a mixture of GPs, however the number of
components of the class would be unknown and this would have to be carefully modelled, perhaps using reversible jump methods \citep{Richardson:1997} or Dirichlet process approaches \citep{Escobar:1995}.
	
\section{Discussion}
This article presents semi-supervised non-parametric Bayesian methods to model spatial proteomics data. Sub-cellular niches display unique signatures along a density-gradient and we exploit this information to construct GP regression models for each niche. The full complement of sub-cellular proteins in then described as mixture of GP regression models, with outliers captured by an additional component in our mixture. This provides cell biologist with a fully Bayesian method to analyse
spatial proteomics data in the non-parametric framework that more closely reflects the biochemical process used to generate the data. This greatly increases model
interpretation and allows us to make more biological sound inferences from our model.

We compared the proposed semi-supervised models to the state-of-the-art model on $5$ different spatial proteomics datasets. Modelling the correlation structure along the density-gradient leads to competitive predictive performance over state-of-the-art models. Empirical Bayes procedures perform either equally well or better than the fully Bayesian approach, at the loss of uncertainty quantification in the hyperparameters. Though this performance improvement should not be over interpreted, since cross-validation assessment is only performed on the labelled data and will not reflect any biased sampling mechanisms that could be at play.

To accelerate computation in our model, we note that the structure of our covariance matrix admits a tensor decomposition, which can be exploited so that fast algorithms for matrix inversion of toeplitz matrices can be employed. These decomposition can then be used to derive formulae for fast computation of the likelihood and gradient of a GP. A stand-alone R-package implementing these methods using high-performance C++ libraries is available at \url{https://github.com/ococrook/toeplitz}. These algorithms and associate formulae are useful to those outside the spatial proteomics community to anyone using GPs with equally spaced observations, even in the unsupervised case. 

We demonstrated that in the presence of labelled data there are two approach to hyperparameter inference. This first is to use empirical-Bayes to optimise the hyperparameters; the other a fully-Bayesian approach, taking into account the uncertainty in these hyperparameters. We propose to use HMC to update these hyperparameters, since highly correlated hyperparameters can induce high autocorrelation and exacerbate issues with random-walk MH updates. We demonstrate
that, in the situation presented here, HMC updates can be up to an order of magnitude more efficient than MH updates.
We further explored the sensitivity of our model to hyper-prior specification, which gives practitioners good default choices.

In two case-studies, we highlighted the value of taking a semi-supervised approach to hyperparameter inference, allowing us to explore the uncertainty in our hyperparameters. In a fully Bayesian approach the uncertainty in the hyperparamters is reflected in the uncertainty of the localisation of proteins to components. Quantifying uncertainty provide cell biologists with a wealth of information to make quantifiable inference about protein sub-cellular localisation.

We plan to disseminate our method via the Bioconductor project \citep{Bioconductor::2004, Huber::2015} and include our code in pRoloc package \citep{pRoloc:2014}. The pRoloc package includes methods for visualisation, processing data and disseminating code in a unified framework. All spatial proteomics data used here is freely available within the Bioconductor package pRolocdata \citep{pRolocdata}.

One potential source of uncertainty in protein localisation is that they can be
residents of multiple sub-cellular compartments. We believe that by proposing
a model which more closely reflects the underlying biochemical rationale for the
experiment we can facilitate models which can infer proteins with multiple locations
with greater confidence. This is the subject of further work.

\section{Supplementary}
\subsection{GP Prior, posteriors and predictive distributions}
Denoting our unknown regression function $\mu_k$, we let observations associated with component $k$ be denoted by $X_k = \{x_1,...,x_{n_k}\}$. Our model tells us that
\begin{equation}
X_k |\boldsymbol{\mu}_k, \sigma_k \sim \mathcal{N}(\boldsymbol{\mu}_k, \sigma^2_kI_D),
\end{equation}
where
\begin{equation}
\boldsymbol{\mu}_k|a_k,l_k \sim GP(0, C_k).
\end{equation}
The posterior distribution follows from normal theory:
\begin{equation}
\boldsymbol{\mu}_k|X_k, \sigma_k \sim GP(\tilde{\boldsymbol{\mu}}_k, \tilde{C}_k),
\end{equation}
where
\begin{align}
\tilde{\boldsymbol{\mu}}_k(t) &= C_k(t,\tau) (C_k(\tau, \tau) + \sigma^2_kI_{n_{k}D})^{-1}X_k(\tau)  \\
\tilde{C}_k(t,t')& = C_k(t,t') - C_k(t,\tau)(C_k(\tau, \tau) + \sigma^2_kI_{n_{k}D})^{-1}C_k(\tau,t').
\end{align}
The mean and covariance functions for the posterior predictive distribution are given by:
\begin{align}
\tilde{\boldsymbol{\mu}}_k(t_{*}) &= C_k(t_{*},\tau) (C_k(\tau, \tau) + \sigma^2_kI_{n_{k}D})^{-1}X_k(\tau)  \\
\tilde{C}_k(t_*,t_*)& = C_k(t_*,t_*) + \sigma_k^2 - C_k(t_{*},\tau)(C_k(\tau, \tau) + \sigma^2_kI_{n_{k}D})^{-1}C_k(\tau,t_*).
\end{align}
\subsection{Derivation of tensor-Toeplitz decomposition inverse}
Let $J_n$ denote a $n\times n$ matrix of ones. Recall that we can write $C$ in the following form:
\begin{equation}
C = \sigma^2I_{nD} + B,
\end{equation}
where
\begin{equation}
B = J_{n} \otimes A,
\end{equation}
and we have denoted $\otimes$ as the Kronecker (tensor) product.
Let $e_n$ denote a column vector of ones of length $n$. It is easy to see that $J_{n} = e_{n}e_{n}^T$. Trivially, we can write $A = I_DA$ and this leads to the following factorisation
\begin{equation}
\begin{split}
B = &\left(e_{n}e_{n}^T\right) \otimes (I_DA). \\
= & (e_{n}\otimes I_D)(e_{n}^T \otimes A),
\end{split}
\end{equation}
where the second equality follows from the mixed-product property of the Kronecker product. Observing that  $e_{n}\otimes I_D$ is a matrix of size $nD \times D$, and $e_{n}^T \otimes A$ is matrix of size $D \times n D$. We thus arrive at the following factorisation:
\begin{equation}
C = \sigma^2I_{nD} + (e_{n}\otimes I_D)I_D(e_{n}^T \otimes A),
\end{equation}
which is in the following form
\begin{equation}
\begin{split}
C& =  M + URV \\
M& = \sigma^2I_{nD},\,\, U =(e_{n}\otimes I_D),\\
R& =   I_D,\,\, V =  e_{n}^T \otimes A.
\end{split}
\end{equation}
Matrices of this form have a simple formula for their inverse (Woodbury Identity):
\begin{equation}
(M + URV)^{-1} = M^{-1} - M^{-1}U(R^{-1} + VM^{-1}U)^{-1}VM^{-1}.
\end{equation}
In our case $R$ is trivially its own inverse and the inverse of $M$ requires only a single computation. Thus the only challenge is to invert $(R^{-1} + VM^{-1}U)$. However, consider the following computations
\begin{equation}
\begin{split}
R^{-1} + VM^{-1}U & = I_D +  (e_{n}^T \otimes A) (\sigma^{-2}I_{nD})(e_{n}\otimes I_D) \\
& = I_D + \sigma^{-2}(e_{n}^Te_{n})\otimes(AI_D)\\
& = I_D + \sigma^{-2} n \otimes A \\
& = I_D + \sigma^{-2} n A
\end{split}
\end{equation}
Recall that $A$ is a  $D \times D$ Toeplitz matrix and so it is easy to see that $R^{-1} + VM^{-1}U$ is also Toeplitz and efficient algorithms exist for inverting them. Denote this inverse by Z and so 
\begin{equation}
\begin{split}
(M + URV)^{-1}& = \sigma^{-2}I_{nD} - \sigma^{-4}(e_{n}\otimes I_D)(Z)(e_{n}^T \otimes A) .\\
& =  \sigma^{-2}I_{nD} - \sigma^{-4}(e_{n}\otimes I_D)(e_{n}^T \otimes ZA) \\
& =  \sigma^{-2}I_{nD} - \sigma^{-4}(e_{n}e_{n}^T ) \otimes (ZA) \\
& =  \sigma^{-2}I_{nD} - \sigma^{-4}J_{n} \otimes (ZA) \\
& =  \sigma^{-2}I_{nD} - \frac{1}{n\sigma^2}J_{n} \otimes (I - Z)
\end{split}
\end{equation}
where the last follows from the following computations, denoting $Z^{-1} = Q$
\begin{equation}
\begin{split}
Q & = I_D + \sigma^{-2}nA \\
Q - \sigma^{-2}nA & = I_D\\
Q^{-1}Q - \sigma^{-2}nQ^{-1}A & = Q^{-1}\\
I_D - \sigma^{-2}nQ^{-1}A & = Q^{-1} \\
I_D - Z & = \sigma^{-2}nZA \\
ZA & = \frac{(I_D - Z)\sigma^2}{n}
\end{split}
\end{equation}
Thus the inversion of $C$ requires only the inversion of a $D\times D$ matrix, which can be performed in $O(D^2)$ computations, this should be compared with a na\"ive inversion of $C$ requiring $O((nD)^3)$ computations, which represents significant savings. We also need the determinant of $C$ and the calculation is straightforward using an elementary determinant lemma.
\begin{equation}
\begin{split}
\det(C) & = \det(M + URV) \\
& = \det(R^{-1} + VM^{-1}U)\det(R)\det(M) \\
& = \det(I_D + ( (e_{n}^T \otimes A) (\sigma^{-2}I_{nD})(e_{n}\otimes I_D)))\det(M) \\
& =  (\sigma^2)^{nD} \det(I_D + \sigma^{-2} n A)
\end{split}
\end{equation}
As before the term in the determinant is Toeplitz and efficient algorithm exists for calculating this determinant.

\subsection{Matrix algorithms}
We state here the require algorithm to invert the covariance matrix $C = \sigma^2I_{nD} + J_n \otimes A$ for a Toeplitz matrix $A$. The algorithms are a minor modification of the algorithms found in \cite{Zhang:2005} to handle the Tensor product. 
\begin{algorithm}
	\caption{Tensor extended Trench algorithm}
	\label{Trench}
	\begin{algorithmic}[1] 
		\Procedure{Trench}{$C = \sigma^2I_{nD} + J_n \otimes A$} \Comment{$C^{-1}$ and $\log \det C$, for Toeplitz A}
		\State $Q \gets I_D + \sigma^{-2}n A $
		\State $q \gets Q_{1,:}^T$
		\State Input $q$ to algorithm 2, returning $v \in \mathbbm{R}^D$ and $l \in \mathbbm{R}^D $
		\State $\bar{Q}(1, 1:D) \gets v(D:1)$
		\State $\bar{Q}(1:D, 1) \gets v(D:1)$
		\State $\bar{Q}(D, 1:D) \gets v(1:D)$
		\State $\bar{Q}(1:D,D) \gets v(1:D)$
		
		\For {$i = 2: \lfloor{(D-1)/2\rfloor} + 1$}
		\For {$j = i: N - i + 1$}
		\State $\bar{Q}(i, j) \gets \bar{Q}(i - 1, j - 1) + \frac{v(D+1-j)v(D+1-i) - v(i-1)v(j-1)}{v(D)}$
		\State $\bar{Q}(j, i) \gets \bar{Q}(i,j)$
		\State $\bar{Q}(N - i + 1, N - j + 1) \gets \bar{Q}(i,j)$
		\State $\bar{Q}(N - j + 1, N - i + 1) \gets \bar{Q}(i,j)$
		\EndFor
		\EndFor
		\State $Z \gets \bar{Q}$
		\State $C^{-1} = \sigma^{-2}I_{nD} - \frac{1}{n\sigma^2}J_{n}^T \otimes (I - Z)$
		\State $\log \det C \gets nD\log(\sigma^2) + l$ 
		\EndProcedure
	\end{algorithmic}
\end{algorithm}

\begin{algorithm}
	\caption{Vector-Inverse and log-determinant algorithm}
	\label{vector-inverse}
	\begin{algorithmic}[1] 
		\Procedure{Vector-Inverse}{$q$} \Comment{$v$ and $l$ as required by algorithm 1}
		\State $\xi \gets \frac{q(2:D)}{q(1)}$
		\State Input $D-1$ and $\xi$ to algorithm 3, returning $z \in \mathbbm{R}^{D-1}$ and $l \in \mathbbm{R}^D $
		\State $l \gets l + D \log q(1)$
		\State $v(D) \gets \frac{1}{(1+\xi^Tz)q(1)}$
		\State $v(1:D-1) \gets v(D)z(D-1:1)$
		\EndProcedure
	\end{algorithmic}
\end{algorithm}

\begin{algorithm}
	\caption{extended Durbin's algorithm }
	\label{Durbin}
	\begin{algorithmic}[1] 
		\Procedure{Durbin}{$m, \xi$} \Comment{$z$ and $l$ as required by algorithm 1}
		\State $z(1) \gets - \xi(1)$
		\State $\beta \gets \alpha \gets 1$
		\State $l \gets 0$
		\For {$i =  1: m-1$}
		\State $\beta \gets (1 -\alpha^2)\beta$
		\State $l = l + \log\beta$
		\State $\alpha \gets \frac{\xi(i+1)+ \xi(i:1)^Tz(1:i)}{\beta}$
		\State $z(1:i) \gets z(1:i) + \alpha z(i:1)$
		\State $z(i+1) \gets \alpha$
		\EndFor
		\State $\beta \gets (1 -\alpha^2)\beta$ 
		\State $l \gets l + \log\beta$
		\EndProcedure
	\end{algorithmic}
\end{algorithm}

\clearpage
\subsection{Derivative of the marginal likelihood}
The derivatives of the marginal likelihood given in equation \ref{equation::GPmarginalliklihood} are given by \citep{Rasmussen:2004}
\begin{equation}\label{equation::GPMLgrad}
\frac{\partial}{\partial\theta_j}\log\left\{p(X_k|\tau,\boldsymbol{\theta_k})\right\} = \frac{1}{2}X_k(\tau)^T \hat{C}^{-1}_k\left(\frac{\partial\hat{C}_k}{\partial \theta_j}\right)\hat{C}^{-1}_k X_k(\tau) - \frac{1}{2}tr\left\{\hat{C}^{-1}_k\left(\frac{\partial \hat{C}_k}{\partial \theta_j}\right)\right\}.
\end{equation}
The partial derivatives of the covariance functions can obtained in a straightforward manner and once evaluated at observations can be structured into blocks just as in equation \ref{equation::covariancestructure}. Letting $\hat{A}_k$ be the diagonal blocks of the covariance matrix in equation \ref{equation::covariancestructure}. The corresponding diagonal blocks of derivative
are given in equation \ref{equation::GPMLpartial}. Blocks not on the diagonal are similar and do not include the derivative with respect to $\theta_3$.
\begin{equation}\label{equation::GPMLpartial}
\begin{split}
\left[\frac{\partial\hat{A}_k}{\partial \theta_1}\right]_{rs} = & a \exp\left\{\left(- \frac{(t_r - t_s)^2}{e^{\theta_1}}\right)\right\}\left(\frac{(t_r - t_s)^2}{e^{\theta_1}}\right) \\
\left[\frac{\partial\hat{A}_k}{\partial \theta_2}\right]_{rs} = & 2e^{2\theta_2}\exp\left(- \frac{(t_r - t_s)^2}{l}\right) \\
\left[\frac{\partial\hat{A}_k}{\partial \theta_3}\right]_{rs} = & 2e^{2\theta_3}\delta_{rs}.
\end{split}
\end{equation}
\subsection{Tensor decompositions for derivatives of the marginal likelihood}
In this appendix we derive formulae for the derivative of the marginal likelihood
exploiting the block structure of our matrices. We first make some preliminary manipulations. We set the following notation $\partial_{\theta_j} = \frac{\partial}{\partial_{\theta_j}}$. First we note that
\begin{equation}
\hat{C}_k^{-1}(\partial_{\theta_j}\hat{C}_k) \hat{C}_k^{-1} = - \partial_{\theta_j}\hat{C}_k^{-1}.
\end{equation}
We recall the following
\begin{equation}
\hat{C}_k^{-1} = \sigma^{-2}I_{nD} - \sigma^{-4} J_n \otimes (ZA),
\end{equation}
and hence the following is true	
\begin{equation}
\partial_{\theta_j}\hat{C}_k^{-1} = \partial_{\theta_j}(\sigma^{-2}I_{nD}) - \partial_{\theta_j}\left\{\sigma^{-4} J_n \otimes (ZA)\right\}.
\end{equation}
We then note that $\partial_{\theta_j}J_n = 0$ and so the following algebraic manipulations hold
\begin{equation}
\begin{split}
\partial_{\theta_j}\left\{J_n \otimes(ZA)\right\} = & \partial_{\theta_j}J_n \otimes (ZA) + J_n \otimes \partial_{\theta_j}(ZA) \\
& = J_n \otimes (\partial_{\theta_j}Z \cdot A + Z \cdot\partial_{\theta_j} A).
\end{split}
\end{equation}
We recall that
\begin{equation}
Z = (I_D + \sigma^{-2}nA)^{-1} = Q^{-1} 
\end{equation}
and so
\begin{equation}
\partial_{\theta_j}Z = - Q^{-1}(\partial_{\theta_j}Q) Q^{-1}.
\end{equation}
It is obvious that
\begin{equation}
\partial_{\theta_j}Q = \partial_{\theta_j}(\sigma^{-2}nA),
\end{equation}
and so 
\begin{equation}
\partial_{\theta_j}Z = - Z\partial_{\theta_j}(\sigma^{-2}nA) Z.
\end{equation}
Whence it follows that
\begin{equation}
\partial_{\theta_j}\hat{C}_k^{-1} = \partial_{\theta_j}(\sigma^{-2}I_{nD}) - \partial_{\theta_j}(\sigma^{-4}) J_n \otimes (ZA) - \sigma^{-4} \left\{- Z\partial_{\theta_j}(\sigma^{-2}nA) ZA + Z\partial_{\theta_j}A\right\}.
\end{equation}
Recall that
\begin{equation}\label{equation::GPMLpartial2}
\begin{split}
\partial_{\theta_1}A_{rs} & =  A_{rs}S_{rs} \\
\partial_{\theta_2}A_{rs} & =  2A_{rs}
\end{split}
\end{equation}
where $S_{rs} = \frac{(t_r - t_s)^2}{l}$. We now derive formulae for the derivatives of the marginal likelihood and we denote $A\odot B$ has the Hadamard (element-wise) product of matrices $A$ and $B$.

\begin{prop}
The derivative of the marginal likelihood in \ref{equation::GPMLgrad} with respect to $\theta_1$ is given by
\begin{equation}
\begin{split}
\partial\theta_1\log\left\{p(X|\tau,\boldsymbol{\theta})\right\} = \frac{1}{2}X(\tau)^T \sigma^{-4}J_n \otimes (ZASZ) X(\tau)  - \frac{1}{2} tr\left(\hat{C}_k^{-1}\partial_{\theta_1}\hat{C}_k\right),
\end{split}
\end{equation}
where
\begin{equation}
tr\left(\hat{C}_k^{-1}\partial_{\theta_1}\hat{C}_k\right) = \sigma^{-2}n\sum_{i}\left(AS\right)_{i,i} - \sigma^{-2}n\sum_{i,j}\left\{(I_D - Z) \odot (AS)\right\}_{ij}.
\end{equation}
\end{prop}
\begin{proof}
We observe the following manipulations, which follow from our preliminary manipulations
\begin{equation}
\begin{split}
\partial_{\theta_1}\hat{C}_k^{-1} & =  \sigma^{-4} (- Z(\sigma^{-2}n\partial_{\theta_1}A) ZA + Z\partial_{\theta_1}A)\\
& =  - \sigma^{-4}J_n \otimes \left\{(Z\partial_{\theta_1}A)(-\sigma^{-2}nZA + I_D)\right\}\\
& = - \sigma^{-4}J_n \otimes \left\{Z (\partial_{\theta_1}A)Z \right\}\\
& = - \sigma^{-4}J_n \otimes (ZASZ),
\end{split}
\end{equation}
where the third line follows from the second because
\begin{equation}
\begin{split}
Q & = I_D + \sigma^{-2}nA \\
Q - \sigma^{-2}nA & = I_D\\
Q^{-1}Q - \sigma^{-2}nQ^{-1}A & = Q^{-1}\\
I_D - \sigma^{-2}nQ^{-1}A & = Q^{-1}. 
\end{split}
\end{equation}
For the trace term, recall that the trace of a product of two matrices is the sum of the Hadamard product of those two matrices. That is
\begin{equation}
tr\left(\hat{C}_k^{-1}\partial_{\theta_j}\hat{C}_k\right) = \sum_{i,j}\left(\hat{C}_k^{-1} \odot\partial_{\theta_j}\hat{C}_k\right)_{i,j}.
\end{equation}
Applying the mixed product property, we see that the following manipulations hold
\begin{equation}
\begin{split}
\hat{C}_k^{-1} \odot\partial_{\theta_1}\hat{C}_k &= \left\{\sigma^{-2}I_{nD} - \sigma^{-4} J_n \otimes (ZA)\right\}\odot \left\{J_n \otimes (AS)\right\}\\
&= \sigma^{-2}I_{nD} \odot \left\{ J_n \otimes (AS)\right\} - \sigma^{-4}\left\{J_n \otimes (ZA)\right\}\odot\left\{J_n \otimes (AS)\right\}\\
&=  \sigma^{-2}I_{nD}diag(AS,AS,\dots,AS) - \sigma^{-4}\left[J_n \otimes \left\{(ZA) \odot (AS)\right\}\right].
\end{split}
\end{equation}
Hence,
\begin{equation}
tr\left(\hat{C}_k^{-1}\partial_{\theta_1}\hat{C}_k\right) = \sigma^{-2}n\sum_{i}\left(AS\right)_{i,i} - \sigma^{-4}n^2\sum_{i,j}\left\{(ZA) \odot (AS)\right\}_{ij}.
\end{equation}
Thus the derivative of the log marginal likelihood is
\begin{equation}
\begin{split}
\partial\theta_1\log\left\{p(X|\tau,\boldsymbol{\theta})\right\} = \frac{1}{2}X(\tau)^T \sigma^{-4}J_n \otimes (ZASZ) X(\tau)  - \frac{1}{2} tr\left(\hat{C}_k^{-1}\partial_{\theta_1}\hat{C}_k\right)
\end{split}
\end{equation}
Then we can substitute $ZA = (I - Z)\frac{\sigma^2}{n}$ to obtain the required result.
\end{proof}
\begin{prop}
	The derivative of the marginal likelihood in \ref{equation::GPMLgrad} with respect to $\theta_2$ is given by
\begin{equation}
\begin{split}
\partial\theta_2\log\left\{p(X|\tau,\boldsymbol{\theta})\right\} = \frac{1}{2}X(\tau)^T \sigma^{-4}J_n \otimes (2ZAZ) X(\tau)  - \frac{1}{2} tr\left(\hat{C}_k^{-1}\partial_{\theta_2}\hat{C}_k\right)
\end{split}
\end{equation}
	where
\begin{equation}
tr\left(\hat{C}_k^{-1}\partial_{\theta_2}\hat{C}_k\right) = 2\sigma^{-2}n\sum_{i}\left(A\right)_{i,i} - \sigma^{-2}n\sum_{i,j}\left\{(I- Z) \odot (2A)\right\}_{ij}.
\end{equation}
\end{prop}
\begin{proof}
As in the previous proposition we observe:
\begin{equation}
\begin{split}
\partial_{\theta_2}\hat{C}_k^{-1} & =  \sigma^{-4} (- Z(\sigma^{-2}n\partial_{\theta_1}A) ZA + Z\partial_{\theta_2}A)\\
& =  - \sigma^{-4}J_n \otimes \left\{(Z\partial_{\theta_2}A)(-\sigma^{-2}nZA + I)\right\}\\
& = - \sigma^{-4}J_n \otimes \left\{Z(\partial_{\theta_2}A)Z\right\}\\
& = - \sigma^{-4}J_n \otimes (2ZAZ).
\end{split}
\end{equation}	
For the trace term, as for $\theta_1$ we proceed as follows
\begin{equation}
\begin{split}
\hat{C}_k^{-1} \odot\partial_{\theta_2}\hat{C}_k &= \left\{\sigma^{-2}I_{nD} - \sigma^{-4} J_n \otimes (ZA)\right\}\odot \left\{J_n \otimes (2A)\right\}\\
&= \sigma^{-2}I_{nD} \odot \left\{J_n \otimes (2A)\right\} - \sigma^{-4}\left\{J_n \otimes (ZA)\right\}\odot\left\{J_n \otimes (2A)\right\}\\
&=  2\sigma^{-2}I_{nD}diag(A,A,\dots,A) - \sigma^{-4}\left[J_n \otimes \left\{(ZA) \odot (2A)\right\}\right].
\end{split}
\end{equation}
Hence,
\begin{equation}
tr\left(\hat{C}_k^{-1}\partial_{\theta_2}\hat{C}_k\right) = 2\sigma^{-2}n\sum_{i}\left(A\right)_{i,i} - \sigma^{-4}n^2\sum_{i,j}\left\{(ZA) \odot (2A)\right\}_{ij}.
\end{equation}
Thus the derivative of the log marginal likelihood is
\begin{equation}
\begin{split}
\partial\theta_2\log\left\{p(X|\tau,\boldsymbol{\theta})\right\} = \frac{1}{2}X(\tau)^T \sigma^{-4}J_n \otimes (2ZAZ) X(\tau)  - \frac{1}{2} tr\left(\hat{C}_k^{-1}\partial_{\theta_2}\hat{C}_k\right)
\end{split}
\end{equation}
Then we can substitute $ZA = (I - Z)\frac{\sigma^2}{n}$ to obtain the required result.
\end{proof}

\begin{prop}
The derivative of the marginal likelihood in \ref{equation::GPMLgrad} with respect to $\theta_1$ is given by
\begin{equation}
\begin{split}
\partial\theta_3\log\left\{p(X|\tau,\boldsymbol{\theta})\right\} = \sigma^{-2} \lVert X(\tau)\rVert_2^2 +  X(\tau)^T J_n \otimes \left\{\frac{(Z^2 - I)}{\sigma^2 n}\right\} X(\tau)  - \frac{1}{2} tr\left(\hat{C}_k^{-1}\partial_{\theta_3}\hat{C}_k\right),
\end{split}
\end{equation}
where
\begin{equation}
tr(\hat{C}_k^{-1}\partial_{\theta_3}\hat{C}_k) = 2nD- 2\sum_i(I - Z)_{ii}.
\end{equation}
\end{prop}
\begin{proof}
We note that $\partial_{\theta_3}\hat{C}_k = 2\sigma^2I_{nD}$ is a scalar multiple of the identity matrix and thus commutes. Hence, we need only compute $\hat{C}_k^{-1}\hat{C}_k^{-1}$ and the trace term. Note the following algebraic manipulations:
\begin{equation}
\begin{split}
\hat{C}_k^{-1}\hat{C}_k^{-1} &  = \left\{\sigma^{-2}I_{nD} - \sigma^{-4} J_n \otimes (ZA)\right\}\left\{\sigma^{-2}I_{nD} - \sigma^{-4} J_n \otimes (ZA)\right\}\\
& = \sigma^{-4}I_{nD}  - 2\sigma^{-6}J_n \otimes (ZA) + \sigma^{-8} (J_nJ_n)\otimes(ZAZA)\\
& = \sigma^{-4}I_{nD}  - 2\sigma^{-6}J_n \otimes (ZA) + n\sigma^{-8} (J_n)\otimes(ZAZA)\\
& = \sigma^{-4}I_{nD}  + J_n \otimes ( - 2\sigma^{-6} ZA + n\sigma^{-8}ZAZA)\\
& = \sigma^{-4}I_{nD}  + J_n \otimes \left\{ \sigma^{-6}(-2I_D + n\sigma^{-2}ZA)ZA\right\}\\
& = \sigma^{-4}I_{nD}  + J_n \otimes \left\{ \sigma^{-6}(-2I_D + I_D - Z)ZA\right\}
\\
& = \sigma^{-4}I_{nD}  + J_n \otimes \left\{ -\sigma^{-6}(I_D + Z)ZA\right\}\\
& = \sigma^{-4}I_{nD}  + J_n \otimes \left\{ -\sigma^{-4}(I_D -Z^2)/n\right\}.
\end{split}
\end{equation}
The compute the trace we note that the following follows directly from the tensor decomposition of $\hat{C}_k^{-1}$:
\begin{equation}
tr(\hat{C}_k^{-1}) = nD\sigma^{-2}- \sigma^{-4}n \sum_i(ZA)_{ii} = nD\sigma^{-2}- \sigma^{-2} \sum_i(I - Z)_{ii}.
\end{equation}
Substituting the formulae shows the desired result it now clear.
\end{proof}
In practice, we never need to compute or even store the full $nD \times nD$ inverse matrix $C^{-1}$, since we can only need to keep track of summaries of the data matrix rather than the full data matrix itself. This is demonstrated in the following proposition.
\begin{prop}
	Let
	\begin{equation*}
	X=
	\begin{bmatrix}
	x_{11} & x_{12} & x_{13} & \dots  & x_{1n} \\
	x_{21} & x_{22} & x_{23} & \dots  & x_{2n} \\
	\vdots & \vdots & \vdots & \ddots & \vdots \\
	x_{D1} & x_{D2} & x_{D3} & \dots  & x_{Dn}
	\end{bmatrix},
	\end{equation*}
	be a $D \times n$ matrix. Let $Y_i = \sum_j X_{i,j}$ be the sum of the $i^{th}$ row of X and written concisely $Y = Xe_n$, where $e_n$ is a $n\times 1$ vector of ones. We write $J_n$ to be the $n\times n$ matrix of ones. Let $R$ be any $D\times D$ matrix. Then the following holds
	\begin{equation}
	vec(X)^T (J_n \otimes R) vec(X) = Y R Y,
	\end{equation}
	where $vec(X)$ denotes the vectorisation of $X$; that is, the $Dn \times 1 $ vector formed by stacking columns of $X$.
\end{prop}
\begin{proof}
	Firstly, observe the following standard algebraic manipulations
	\begin{equation}
	\begin{split}
	(J_n \otimes R) vec(X) & = vec(RXJ_n)\\
	& = vec(RXe_ne_n^T)\\
	& = vec(RYe_n^T) \\
	& = (e_n \otimes R) vec(Y) \\
	& = (e_n \otimes R) Y.
	\end{split}
	\end{equation}
	Thus, using the above, it follows that
	\begin{equation}
	\begin{split}
	vec(X)^T (J_n \otimes R) vec(X) & = vec(X)^T (e_n \otimes R) Y\\
	& = vec(R^TXe_n)^T Y\\
	& = vec(R^TY)^T Y \\
	& = (R^TY)^TY \\
	& = Y^TRY,
	\end{split}
	\end{equation}
	as required. 
\end{proof}

\subsection{Hamiltonian Monte-Carlo for GP hyperparameters}
The Hamiltonian can be decomposed into potential and kinetic energies $H(\bm{x},\bm{p}) = U(\bm{x}) + K(\bm{p})$. The canonical distribution is then given by:
\begin{equation}
p(\bm{x},\bm{p}) \propto \exp(-H(\bm{x},\bm{p})) \propto p(\bm{x})p(\bm{p}).
\end{equation}
The distribution of momentum component is chosen as a Gaussian distribution with diagonal covariance matrix $M = diag(m_1,...,m_r)$ and thus the distribution and kinetic energies are given by
\begin{equation}
\begin{split}
p(\bm{p}) &= N(0,M)\\
K(\bm{p}) & = \frac{\bm{p}M^{-1}\bm{p}^T}{2} \\
\nabla K & = M^{-1}\bm{p}.
\end{split}
\end{equation}
It is easy to see from the canonical distribution that $U(\bm{x}) = - \log(p(\bm{x}))$ is the required choice for the potential. In practice, we need to simulate from Hamiltonian dynamics. Hamilton's equations are given by a coupled system:
\begin{equation}
\begin{split}
\frac{d\bm{p}}{dt}& = -\nabla_{\bm{x}} H(\bm{x},\bm{p})\\
\frac{d\bm{x}}{dt}& = \nabla_{\bm{p}} H(\bm{x},\bm{p}).
\end{split}
\end{equation}
Such a system is called symplectic and thus a numerical schema which is a symplectic integrator is required to simulate the required dynamics \citep{Neal:2011}. The leapfrog algorithm is the standard choice \citep{Mackay:2003}. This algorithm does not exactly conserve energy and so a Metropolis accept/reject step is required is remove the induced bias \citep{Beskos:2013}. An MCMC algorithm can then be constructed to sample from the required distribution, where proposals are made using Hamiltonian evolutions. Recall, we are required to simulate the Hamiltonian evolutions. To simulate an evolution over time $T$, take $L$ steps of size $\delta$ such that $L\delta \geq T$. One step of the leapfrog algorithm of size $\delta$ for Hamilton's dynamics starting at time $t$ is given by the following
\begin{equation}\label{leapfrog}
\begin{split}
\bm{p}(t+\delta/2) &= \bm{p}(t) - \frac{\delta}{2} \nabla U \bm{x}(t)\\
\bm{x}(t+\delta) & = \bm{x}(t) + \delta \nabla K \bm{p}(t+\delta/2) \\
\bm{p}(t+\delta) &= \bm{p}(t + \delta/2) - \frac{\delta}{2} \nabla U \bm{x}(t+ \delta)
\end{split}
\end{equation}
We can now summarise the HMC algorithm to sample $n$ samples from a target distribution $p(\bm{x})$.
\begin{enumerate}
	\item Set $t = 0$
	\item Sample a position value from the prior $\bm{x}_0 \sim p_0$
	\item Do until $t = n$
	\begin{enumerate}
		\item Set $t = t + 1$
		\item Sample an initial momentum variable $\bm{p}_0 \sim p(\bm{p})$
		\item Set $\bm{x}_0 = \bm{x}_{t-1}$
		\item Run algorithm \ref{leapfrog} for $L$ step of size $\delta$ and obtain proposal states $\bm{x}_{*}$ and $\bm{p}_{*}$
		\item Compute the Metropolis ratio 
		\begin{equation}
		\Lambda = \exp(- (U(\bm{x}_{*})+K(\bm{p}_{*})) + (U(\bm{x}_{0})+K(\bm{p}_{0})))
		\end{equation}
		\item Sample $u \sim U[0,1]$ if $\Lambda > u$ set $\bm{x}_t = \bm{x}_{*} $, else $\bm{x}_t = \bm{x}_{t-1} $
	\end{enumerate}
\end{enumerate}
We can now specify the details for sampling the hyperparameters of a Gaussian Process with standard normal hyperpriors. Using a squared exponential covariance function and re-parametrising, as before, we first specify our target distribution $p(\bm{x}) = p(\boldsymbol{\theta}|X(\tau)) \propto p(X(\tau)|\boldsymbol{\theta})p_0(\boldsymbol{\theta})$. Now considering \begin{equation}
U(\bm{x}) = - \log(p(\bm{x})) = - \log(p(X(\tau)|\boldsymbol{\theta})) - \log(p_0(\boldsymbol{\theta})) + constant,
\end{equation}
the first term can be computed by marginalising and is recognised as the marginal likelihood given in equation $\ref{equation::GPmarginalliklihood}$. Recalling that we have a standard normal prior the negative log prior and its gradient is given by is given by
\begin{equation}
\begin{split}
- \log(p_0(\boldsymbol{\theta})) & = \frac{3}{2}\log((2\pi)) + \frac{\bm{\theta} \bm{\theta} ^T}{2}\\
\nabla(- \log(p_0(\boldsymbol{\theta})) ) &=\bm{\theta} 
\end{split}
\end{equation}
where $\bm{\theta}  = (\theta_1, \theta_2, \theta_3).$ Hence, we can write down the gradient of the potential energy using the above and equation $\ref{equation::GPMLgrad}$. We further reintroduce the dependence on $k$,
\begin{equation}
\begin{split}
\nabla U(\bm{x}) = \nabla (-\log(p(\bm{x}))) = &\frac{1}{2}tr\left(\left(\hat{C}_k^{-1} -\alpha\alpha^T \right)\nabla \hat{C}_k\right)  + \bm{x} \\
\alpha & = \hat{C}_k^{-1}X_k(\tau)
\end{split}
\end{equation}
We recall that $\nabla \hat{C}_k$ can be computed from equation \ref{equation::GPMLpartial}. Thus we have everything we need to simulate Hamiltonian dynamics to explore our target distribution. In practice, we make a few standard adaptations to the above algorithm as detailed in \citep{Neal:2011}. We sample $\delta$ from a uniform distribution on $\mathcal{U}\left[a,b\right]$, as well as using a partial momentum refreshment with parameter $\alpha$. More specifically, given $\bm{p}$ from the previous iteration of the HMC algorithm and a sample $n \sim p_0(\bm{p})$ set $\bm{p}'$ as
\begin{equation}
\bm{p}' = \alpha \bm{p} + (1-\alpha^2)^{1/2}n.
\end{equation}
\subsection{Assessing Convergence}
To assess convergence of our MCMC algorithms we visualise trace plots of our chains. In addition, we run two chains in parallel and we look at the potential scale reduction factors (the $\hat{R}$ statistic) and their upper $95\%$ confidence limits \citep{Gelman:1992,Brooks:1998} using the coda R package \citep{CODA}. We note that values of the $\hat{R}$ statistics far from $1$ indicate non-convergence. For the \textit{Drosophila} example we run our a Hamiltonian-within-Gibbs sampler for $20,000$ iterations, performing a Hamiltonian Monte Carlo move to update the hyperparameters every $50$ iterations. We monitor the hyperparameters of the GPs carefully, since correlations can lead to slow exploration. As a representative example, we plot, in figure \ref{figure:parallelchainsTan}, the two parallel chains for the length-scale of GP associated to the ER. We see that mixing is rapidly achieved and that the upper $95\%$ confidence limit of $\hat{R} \approx 1$, indicating convergence. We contrast this with using a Metropolis-within-Gibbs sampler, in which we perform and Metropolis-Hastings move every $10$ iterations. We see from figure \ref{figure:parallelchainsMHtan} the random walk nature of the parameters. In this case the upper $95\%$ confidence limit of $\hat{R} \approx 1.05$, thus our chain has most likely converged but exploration of the probability space is clearly slow. We also asses convergence using parallel chains in the set-up of section \ref{section:mouse}. We monitor $\epsilon$ the mixing weight of the outlier component, as an example, which can be seen in figure \ref{figure:parallelchainsAndy}. The upper $95\%$ confidence limit of $\hat{R} \approx 1.01$, indicating convergence.
\begin{figure}[h]
\centering
\includegraphics[width = 10cm]{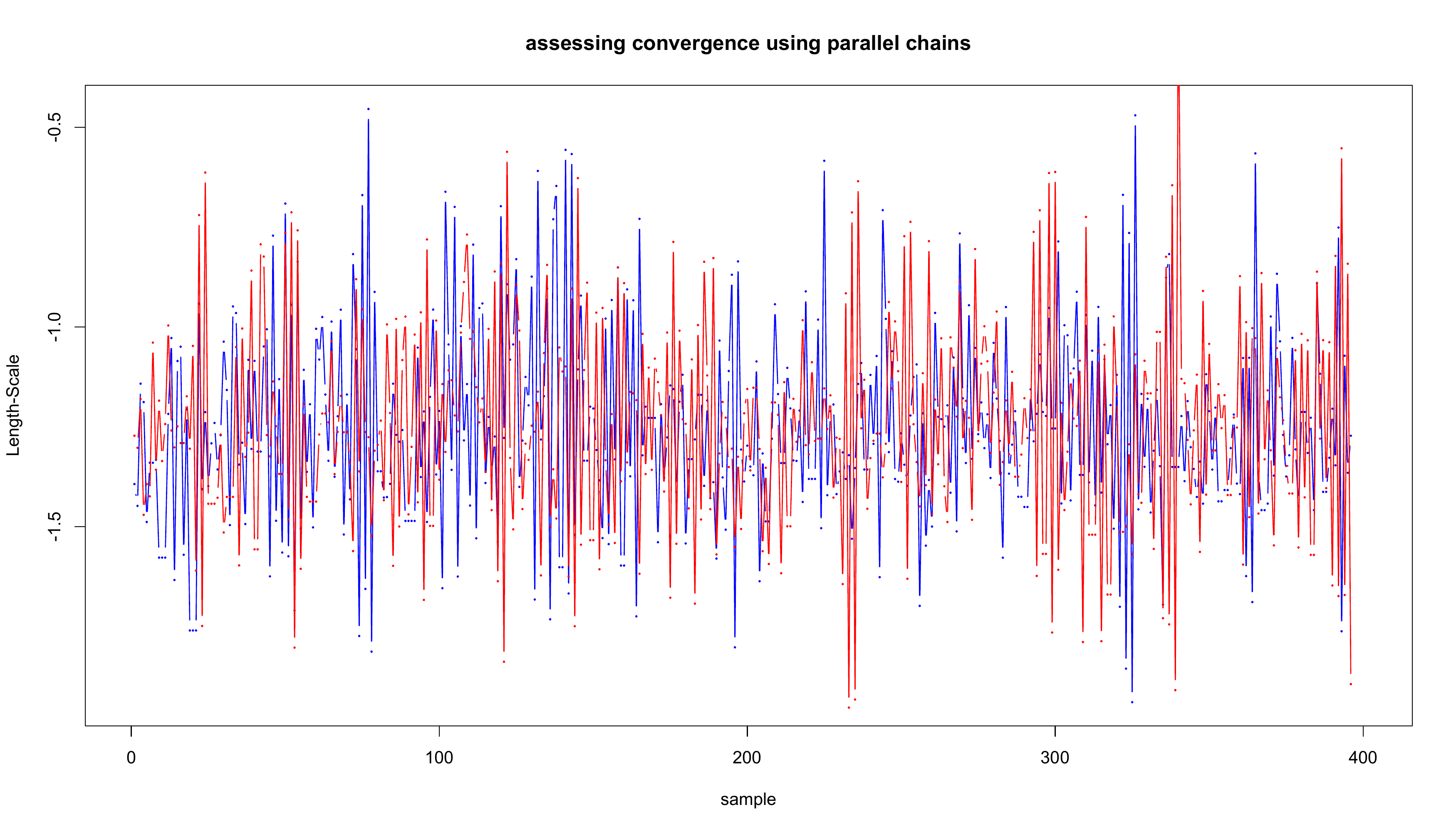}
\caption{Two parallel chains for the length-scale of the Gaussian process regression model associated to the endoplasmic reticulum, using a Hamiltonian-within-Gibbs sampler}
\label{figure:parallelchainsTan}
\end{figure}

\begin{figure}[h]
	\centering
	\includegraphics[width = 10cm]{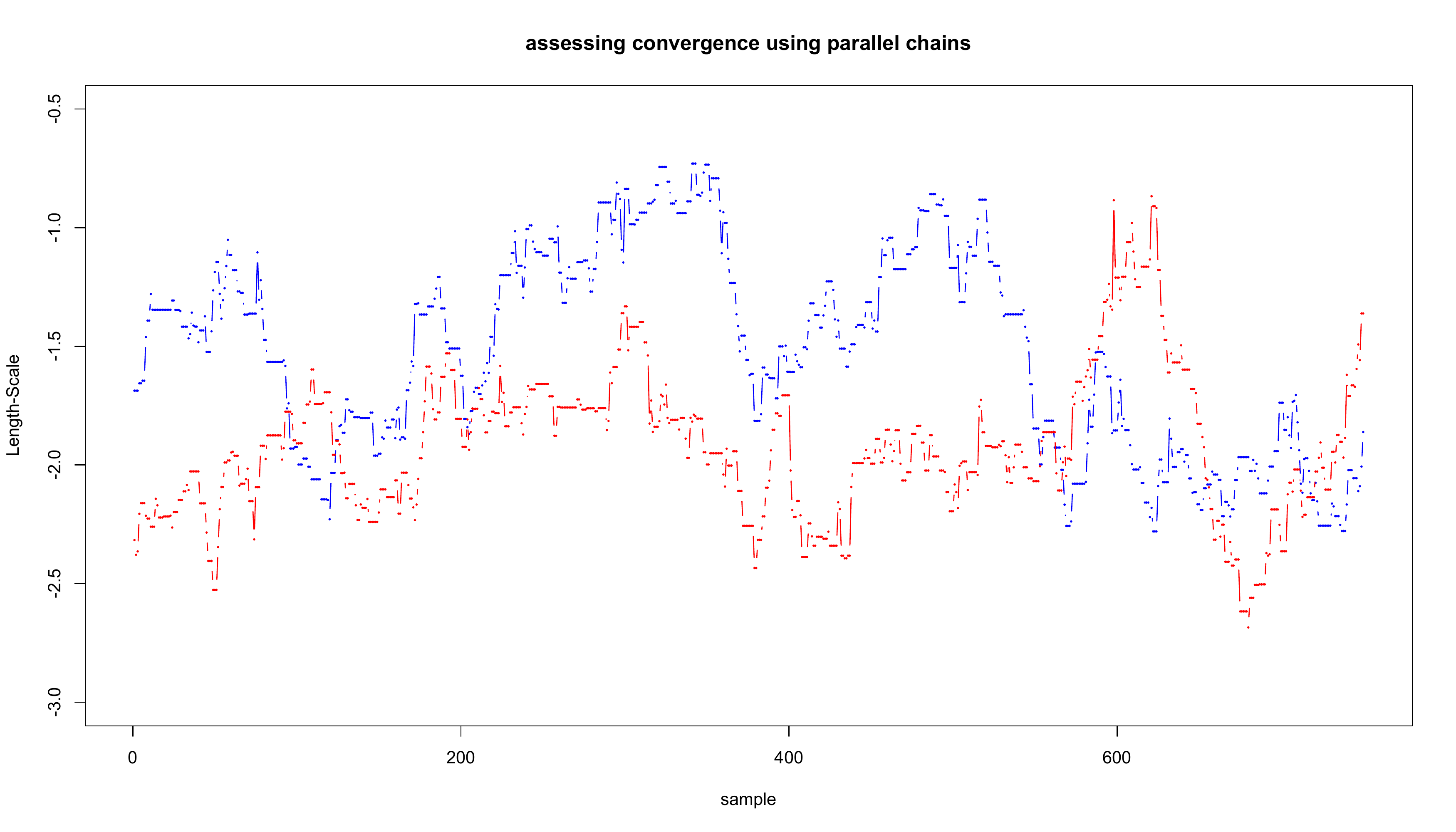}
	\caption{Two parallel chains for the length-scale of the Gaussian process regression model associated to the endoplasmic reticulum, using a Metropolis-within-Gibbs sampler}
	\label{figure:parallelchainsMHtan}
\end{figure}
\begin{figure}[h]
\centering
\includegraphics[width = 10cm]{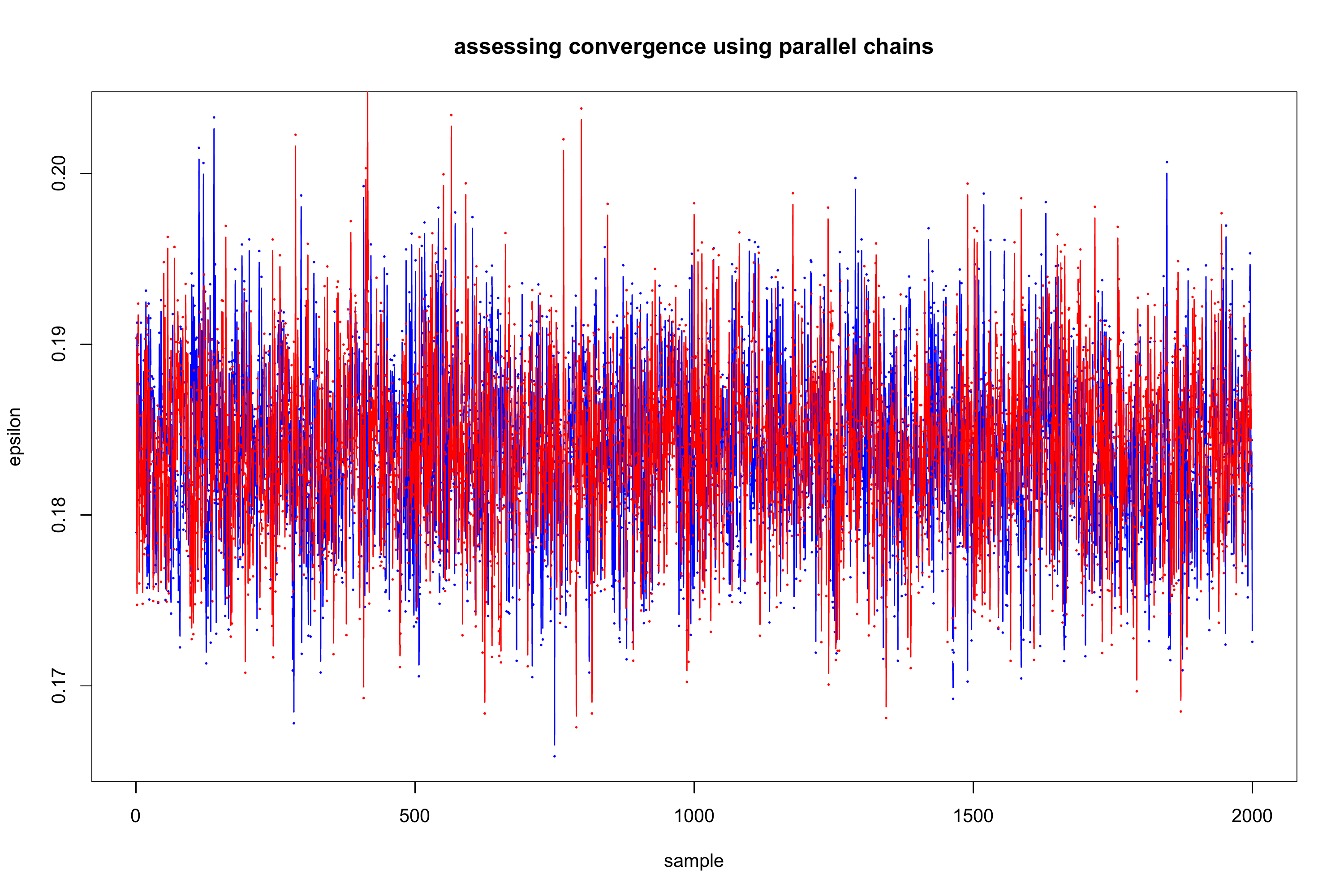}
\caption{Two parallel chains for $\epsilon$, the weight of the outlier component in the mixture model}
\label{figure:parallelchainsAndy}
\end{figure}
\clearpage
\subsection{Tables of hyperparameters}
Tables of hyperparameters and hyperparameter distributions for the mouse pluripotent stem cell data.
\bigskip
\begin{table}
	\caption{A table of log hyperparameters for a GP found by optimising the
	marginal likelihood using L-BFGS}		\label{table:1}
	\centering
	\fbox{%
	\begin{tabular}{c c c c }
		\hline
		Sub-cellular niche & Length-scale & Amplitude & Noise \\ [0.5ex]
		\hline
		40S Ribosome & 0.81 & -2.45 & -4.23 \\ 
		60S Ribosome & 0.61 & -2.90 & -4.28 \\ 
		Actin cytoskeleton & 0.44 & -2.67 & -3.77 \\ 
		Cytosol & 0.80 & -2.17 & -3.66 \\ 
		ER/Golgi apparatus & 0.96 & -2.60 & -3.82 \\ 
		Endosome & 0.48 & -2.48 & -3.49 \\ 
		Extracellular matrix & 0.53 & -2.74 & -4.06 \\ 
		Lysosome & 0.64 & -2.43 & -4.03 \\ 
		Mitochondrion & 0.55 & -2.26 & -3.77 \\ 
		Nucleus - Chromatin & 0.46 & -2.23 & -3.71 \\ 
		Nucleus - Non-chromatin & 0.23 & -2.25 & -3.47 \\ 
		Peroxisome & 0.78 & -2.40 & -3.78 \\ 
		Plasma membrane & 0.28 & -2.41 & -3.92 \\ 
		Proteasome & 0.70 & -2.01 & -4.16 \\
		\hline
\end{tabular}}
\end{table}

\begin{table}
	\caption{A table of log GP hyperparameters with $95\%$ equi-tailed credible intervals summarised from samples produced using HMC}
	\centering
	\fbox{%
\begin{tabular}{c c c c }
	\hline
	& Length-scale & Amplitude & Noise \\ [0.5ex]
	\hline
	40S Ribosome & $0.54 \left[-0.64, 1.08\right]$ & $-2.39 \left[-2.74, -2.01\right]$  &$ -4.23 \left[-4.29, -4.17\right]$ \\ 
	60S Ribosome &$ 0.51 \left[-0.20, 0.93\right]$ & $-2.77\left[-3.18, -2.31\right]$ & $-4.28 \left[-4.31, -4.23 \right]$ \\ 
	Actin cytoskeleton & $0.33 \left[-0.52, 0.81\right]$ & $-2.55\left[-2.89, -2.20\right]$ & $-3.76\left[-3.84, -3.68\right]$ \\ 
	Cytosol & $0.69 \left[- 0.01, 1.11\right]$ & $-2.04 \left[-2.43, - 1.60\right]$ & $-3.66 \left[-3.70, - 3.61\right]$ \\ 
	ER/Golgi apparatus & $0.89 \left[0.29, 1.37\right] $ & $-2.53 \left[ -2.90, - 1.89\right]$ & $-3.82 \left[-3.85, -3.79\right]$ \\ 
	Endosome & $0.39 \left[-0.24, 0.84\right]$ & $-2.37 \left[-2.68, - 1.92\right]$ & $-3.48 \left[-3.58, - 3.39\right]$ \\ 
	Extracellular matrix & $0.37 [-0.32,0.92]$ & $-2.65 [-2.97,-2.24]$ & $-4.05 \left[-4.14, -3.96\right]$ \\ 
	Lysosome & $0.54 \left[-0.31, 0.94\right]$ & $-2.36 \left[-2.69, -2.00\right]$ & $-4.03 \left[ -4.09, - 3.98\right]$ \\ 
	Mitochondrion & $0.53 \left[0.12, 0.95\right]$ & $-2.12 \left[-2.38, -1.80\right]$ & $-3.77 \left[ -3.78, -3.75\right]$ \\ 
	Nucleus - Chromatin & $0.46 \left[0.05,0.86\right]$ & $-2.14 \left[-2.45, -1.81\right]$ & $-3.71 \left[-3.75, -3.68\right]$ \\ 
	Nucleus - Non-chromatin & $0.05\left[-1.19, 0.69\right]$ & $-2.09 \left[-2.48, -1.71\right]$ & $-3.47 \left[ -3.50, -3.44\right]$ \\ 
	Peroxisome & $0.75 \left[0.28, 1.17\right]$ & $-2.31\left[-2.62, -1.92\right] $ & $-3.78 \left[-3.85, -.3.69\right]$ \\ 
	Plasma membrane & $0.02 \left[-1.03, 0.67\right]$ & $-2.32\left[-2.65, -1.91\right]$ & $-3.91 \left[-3.95, - 3.86\right]$ \\ 
	Proteasome & $0.59 \left[0.16. 0.97\right] $ & $-1.94 \left[-2.26, - 1.52\right]$ & $-4.15 \left[-4.21, - 4.10\right]$ \\ 
	\hline
\end{tabular}}
\end{table}
\newpage
\subsection{Prior Specifications}
 The priors for the comparison between classifiers in section \ref{section:performance} are as follows. The normal-inverse-Wishart prior for the multivariate Gaussian distributions was the following: the mean was set as the empirical mean of whole data, the shrinkage was set to $0.01$, the degrees of freedom was set to be the number of variables plus 2, the scale matrix was set to the identity matrix. For the GP we placed standard normal prior on each log hyperparameter. For all methods the Beta prior for the outlier component prior weight was set to be $\mathcal{B}(2,10)$ and the mixing proportions for each component was given symmetric Dirichlet prior with $\alpha = 1$.

\newpage
\bibliographystyle{apalike}
\bibliography{GPBayesProt}
\end{document}